\documentclass[11pt]{article}
\usepackage[utf8]{inputenc}
\usepackage{thm-restate}
\usepackage{authblk}
\usepackage{verbatim}
\usepackage[top=1in, bottom=1in, left=1.2in, right=1.2in]{geometry}
\usepackage{enumerate}
\usepackage{hyperref}
\usepackage[para]{threeparttable}
\usepackage{titling}
\hypersetup{
    colorlinks,
    linkcolor={blue},
    citecolor={blue},
    urlcolor={blue}
}

\usepackage{amsmath, amsthm, amssymb}
\usepackage{booktabs}
\usepackage{graphicx}
\usepackage{algorithm}%
\usepackage[kw]{pseudo}
\usepackage{float}
\usepackage{capt-of}
\usepackage{isabelle}
\usepackage{needspace}
\usepackage{longtable}
\isabellestyle{it}
\newcommand{\eps}{\varepsilon}
\newcommand{\bigo}{\mathcal O}
\newcommand{\size}[1]{\left\lvert#1\right\rvert}
\newcommand{\floor}[1]{\left\lfloor#1\right\rfloor}
\newcommand{\ceil}[1]{\left\lceil#1\right\rceil}

\newtheorem{theorem}{Theorem}[section]
\newtheorem{lemma}[theorem]{Lemma}
\theoremstyle{definition}

\DeclareMathOperator{\prob}{\mathcal P}
\DeclareMathOperator{\ld}{ld}
\DeclareMathOperator{\expect}{\mathbb E}
\DeclareMathOperator{\var}{\mathbb V}
\DeclareMathOperator{\poly}{poly}
\newcommand{\isatheory}[1]{\texttt{\detokenize{#1}}}
\newcounter{consts}
\newcommand{\makeconst}[1]{\refstepcounter{consts}C_{\theconsts}\label{#1}}
\newcommand{\const}[1]{C_{\ref{#1}}}
\title{An embarrassingly parallel optimal-space cardinality estimation algorithm}

\author{Emin Karayel\thanks{E-mail: \href{mailto:me@eminkarayel.de}{me@eminkarayel.de}}}
\affil{Department of Informatics, Technische Universität München, Germany}
\predate{}
\postdate{}
\date{}

\begin{document}
\maketitle

\begin{abstract}
In 2020 B\l{}asiok (ACM Trans.\ Algorithms 16(2) 3:1-3:28) constructed an optimal space streaming algorithm for the cardinality estimation problem with the space complexity of $\bigo(\eps^{-2} \ln(\delta^{-1}) + \ln n)$ where $\eps$, $\delta$ and $n$ denote the relative accuracy, failure probability and universe size, respectively.
However, his solution requires the stream to be processed sequentially.
On the other hand, there are algorithms that admit a merge operation; they can be used in a distributed setting, allowing parallel processing of sections of the stream, and are highly relevant for large-scale distributed applications.
The best-known such algorithm, unfortunately, has a space complexity exceeding $\Omega(\ln(\delta^{-1}) (\eps^{-2} \ln \ln n + \ln n))$.
This work presents a new algorithm that improves on the solution by B\l{}asiok, preserving its space complexity, but with the benefit that it admits such a merge operation, thus providing an optimal solution for the problem for both sequential and parallel applications.
Orthogonally, the new algorithm also improves algorithmically on B\l{}asiok's solution (even in the sequential setting) by reducing its implementation complexity and requiring fewer distinct pseudo-random objects.
\end{abstract}
\section{Introduction}
In 1985 Flajolet and Martin~\cite{flajolet1985} introduced a space-efficient streaming algorithm for the estimation of the count of distinct elements in a stream $a_1,...,a_m$ whose elements are from a finite universe $U$.
Their algorithm does not modify the stream, observes each stream element exactly once and its internal state requires space logarithmic in $n=\size{U}$.
However, their solution relies on the model assumption that a given hash function can be treated like a random function selected uniformly from the family of all functions with a fixed domain and range.
Despite the ad-hoc assumption, their work spurred a large number of publications\footnote{Pettie and Wang~\cite[Table 1]{seth2021} summarized a comprehensive list.}, improving the space efficiency and runtime of the algorithm.
In 1999 Alon et al.~\cite{alon1999} identified a solution that avoids the ad-hoc model assumption.
They use $2$-independent families of hash functions, which can be seeded by a logarithmic number of random bits in $\size{U}$ while retaining a restricted set of randomness properties.
Their refined solution was the first rigorous Monte-Carlo algorithm for the problem.
Building on their work, Bar-Yossef et al.\ in 2002~\cite{baryossef2002}, then Kane et al. in 2010~\cite{kane2010} and lastly, B\l{}asiok in 2020~\cite{blasiok2020}\footnote{An earlier version of B\l{}asiok's work was presented in the ACM-SIAM Symposium on Discrete Algorithms in 2018.~\cite{blasiok2018}} developed successively better algorithms achieving a space complexity of $\bigo(\eps^{-2} \ln(\delta^{-1}) + \ln n)$, which is known to be optimal~\cite[Theorem 4.4]{jayram2013}.
\begin{table}[h]%
\begin{threeparttable}[b]
\begin{tabular}{l l l}
\toprule
Year, Author & Space Complexity & Merge \\
\midrule
1981, Flajolet and Martin & $\bigo(\eps^{-2} \ln n)$ for constant $\delta$ \tnote{a} & Yes \\
1999, Alon et al. & $\bigo(\ln \ln n)$ for $\delta = 2 (\eps+1)^{-1}$ & Yes \\
2002, Bar-Yossef et al.\tnote{b} & $\bigo(\ln(\delta^{-1})(\eps^{-2} \ln \ln n + \poly(\ln(\eps^{-1}), \ln \ln n) \ln n))$ \tnote{c} & Yes \\
2010, Kane et al. & $\bigo(\ln(\delta^{-1})(\eps^{-2} + \ln n))$ & No \\
2020, B\l{}asiok & $\bigo(\ln(\delta^{-1})\eps^{-2} + \ln n)$ & No \\
This work & $\bigo(\ln(\delta^{-1})\eps^{-2} + \ln n)$ & Yes \\
\bottomrule
\end{tabular}
\begin{tablenotes}
\item [a] Random oracle model.
\item [b] Algorithm 2 from the publication.
\item [c] The notation $\poly(a,b)$ stands for a term polynomial in $a$ and $b$. 
\end{tablenotes}
\end{threeparttable}
\caption{Important cardinality estimation algorithms.}
\label{tab:algorithms}
\end{table}%
These algorithms return an approximation $Y$ of the number of distinct elements $\size{A}$ (for $A := \{a_1,\ldots,a_m\}$) with relative error $\eps$ and success probability $1 - \delta$, i.e.:
\[
    \prob(\size{Y-\size{A}} \leq \eps \size{A}) \geq 1-\delta
\]
where the probability is only over the internal random coin flips of the algorithm but holds for all inputs.

Unmentioned in the source material is the fact that it is possible to run the older algorithms by Alon et al.\ and Bar-Yossef et al.\ in a parallel mode of operation.
This is due to the fact that the algorithms make the random coin flips only in a first initialization step, proceeding deterministically afterwards and that the processing step for the stream elements is commutative.
For example, if two runs for sequences $a$ and $b$ of the algorithm had been started with the same coin flips, then it is possible to introduce a new operation that merges the final states of the two runs and computes the state that the algorithm would have reached if it had processed the concatenation of the sequences $a$ and $b$ sequentially.
Note that the elements of the sequences are not required to be disjoint.
This enables processing a large stream using multiple processes in parallel.
The processes have to communicate at the beginning and at the end to compute an estimate.
The communication at the beginning is to share random bits, and the communication at the end is to merge the states.
Because there is no need for communication in between, the speed-up is optimal with respect to the number of processes, such algorithms are also called embarrassingly parallel~\cite[Part 1]{foster1995}.
This mode of operation has been called the distributed streams model by Gibbons and Tirthaputra~\cite{gibbons2001}.
Besides the distributed streams model, such a merge operation allows even more varied use cases, for example, during query processing in a Map-Reduce pipeline~\cite{dean2010}\ or as decomposable/distributive aggregate functions within OLAP cubes~\cite{han2012}. Figure~\ref{fig:dist_stream_example} illustrates two possible modes of operation (among many) enabled by a merge function.

\begin{figure}[ht!]
\centering

\includegraphics{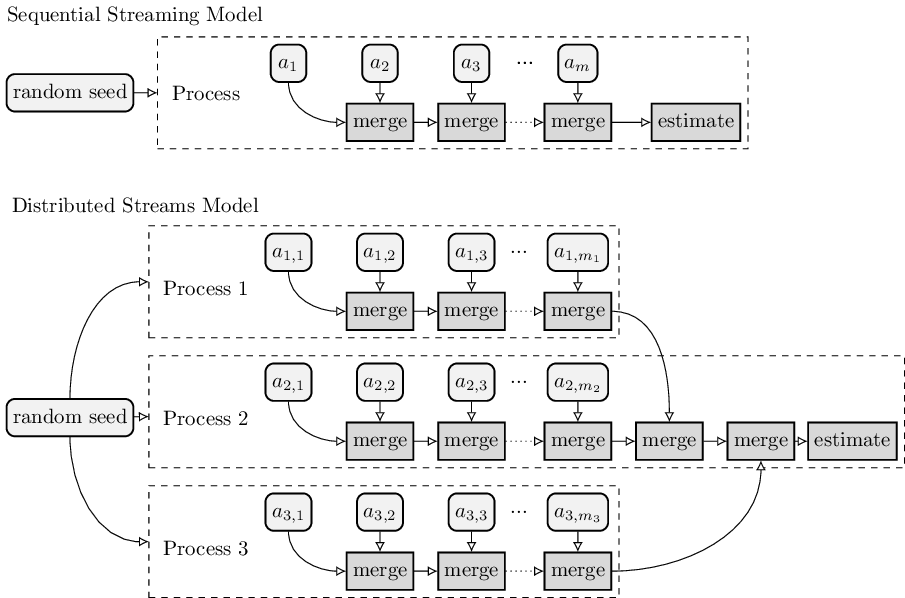}

\caption{Example use cases for cardinality estimation algorithms that support merge.}
\label{fig:dist_stream_example}
\end{figure}

However, an extension with such a merge operation is not possible for the improved algorithms by Kane et al.\ and B\l{}asiok.
This is because part of their correctness proof relies inherently on the assumption of sequential execution, in particular, that the sequence of states is monotonically increasing, which is only valid in the sequential case.
%The latter enables a union-bound on the possible number of times certain kinds of state-changes can happen.
This work introduces a new distributed cardinality estimation algorithm which supports a merge operation with the same per-process space usage as the optimal sequential algorithm by B\l{}asiok: $\bigo(\eps^{-2} \ln(\delta^{-1}) + \ln n)$.
Thus the algorithm in this work has the best possible space complexity in \emph{both} the sequential and distributed streaming model.\footnote{That the complexity is also optimal for the distributed setting is established in Section~\ref{sec:counter_example}.} (Table~\ref{tab:algorithms} provides a summary of the algorithms mentioned here.)%

The main idea was to modify the algorithm by B\l{}asiok into a history-independent algorithm. This means that the algorithm will, given the same coin-flips, reach the same state independent of the order in which the stream elements arrive, or more precisely, independent of the execution tree as long as its nodes contain the same set of elements.
This also means that the success event, i.e., whether an estimate computed from the state has the required accuracy, only depends on the set of distinct stream elements encountered (during the execution tree) and the initial random coin flips.
As a consequence and in contrast to previous work, the correctness proof does not rely on bounds on the probability of certain events over the entire course of the algorithm, but can be established independent of past events.

B\l{}asiok uses a pseudo-random construction based on hash families, expander walks, an extractor based on Parvaresh-Vardy codes~\cite{guruswami2009} and a new sub-sampling strategy~\cite{blasiok2020}[Lem.~39]. I was able to build a simpler stack that only relies on hash families and a new two-stage expander graph construction, for which I believe there may be further applications. To summarize --- the solution presented in this work has two key improvements:
\begin{itemize}
\item Supports the sequential \emph{and} distributed streaming model with optimal space.
\item Requires fewer pseudo-random constructs, i.e., only hash families and expander walks.
\end{itemize}
In the next Section I will briefly discuss the history of the algorithm by B\l{}asiok, because it is best understood as a succession of improvements starting from Alg. 2 by Bar-Yossef et al. In this context it will be possible to introduce the improvements in the new algorithm in more detail. After that, I present new results on expander walks (Section~\ref{sec:chernoff}) needed in the new pseudo-random construction and a self-contained presentation of the new algorithm and its correctness proof (Sections~\ref{sec:alg} and \ref{sec:ext_arb}). 
Concluding with a discussion of its optimality in the distributed setting (Section~\ref{sec:counter_example}), its runtime complexity (Section~\ref{sec:runtime}) and a discussion of open research questions (Section~\ref{sec:conclusion}).

The results obtained in this work have also been formally verified~\cite{Distributed_Distinct_Elements-AFP} using the proof assistant Isabelle~\cite{nipkow2002isabelle}. 
Isabelle has been used to verify many~\cite{afp} advanced results from mathematics (e.g. the prime number theorem~\cite{Prime_Number_Theorem-AFP}) and computer science (e.g. the Cook-Levin theorem~\cite{Cook_Levin-AFP}). % for example the Prime Number Theorem~\cite{Prime_Number_Theorem-AFP} or Edmunds-Karp maximal flow algorithm~\cite{EdmondsKarp_Maxflow-AFP}.
For readers mainly interested in the actual results, the formalization can be ignored as the theorems and lemmas all contain traditional mathematical proofs.
Nevertheless, Table~\ref{tab:formalization} references the corresponding formalized fact for every lemma and theorem in this work.

\section{Background\label{sec:background}}
The algorithm in this work is a refinement of the solution by B\l{}asiok, Kane et al., and Alg.\ 2 by Bar-Yossef et al. This section introduces them briefly and describes the key improvements in this work, whenever the relevant concepts are introduced.
Bar-Yossef's algorithm relies on the fact that when, $r$ balls are thrown randomly and independently into $b$ bins, the expected number of bins hit by at least one ball will be 
\begin{equation}
\label{eq:hit_balls}
    b \left(1-\left(1-b^{-1}\right)^r\right) \textrm{.}
\end{equation}
If $r$ --- the number of balls --- is close to $b$ it is possible to invert Eq.~\ref{eq:hit_balls} to obtain an estimate for $r$ by counting the number of hit bins. Bar-Yossef et al.\ were able to show that its possible to choose the bins for each ball $k$-wise independently, where $k$ is in $\bigo(\ln(\eps^{-1}))$ instead of completely independently, because the expectation and variance of the number of hit bins coverges exponentially fast to the corresponding values for the idealized case of independently choosing a random bin for each ball.
With that in mind we can imagine an algorithm choosing a hash function $g$ randomly from a $k$-wise independent hash family from the universe $U=[n]$ to $[b]$ that maps each stream element using $g$ into the $b$ bins and tracks whether a bin was `hit' by a stream element.
To be able to deal with a situation where the number of distinct stream elements is much larger than $b$ they introduce a sub-sampling strategy. This works by choosing a second pairwise-independent hash function $f$ with a geometric distribution, i.e., the universe elements are assigned a level, where each universe element has a level $\geq 0$. Only half of them have a level $\geq 1$ and only a quarter of them have a level $\geq 2$, etc.
They then choose to restrict the analysis to the universe elements with a given minimum level --- the sub-sampling threshold --- such that the cardinality of the stream elements in that part of the universe is close to the number of bins. To achieve that they not only store for a bin, whether a stream element was hashed into it, but also the maximal level of the stream elements mapping into it. This is the reason for the $\ln \ln n$ factor in the term $\eps^{-2} \ln \ln n$ in the space complexity of their algorithm.
They also run in parallel a second rough estimation algorithm. At the end they use the rough estimation algorithm to get a ball park for the number of stream elements and determine a sub-sampling threshold $s$, for which the number of stream elements is expected to be approximately the number of bins. They then count the number of bins hit by at least one stream element of that level or higher. By inverting Eq.~\ref{eq:hit_balls} it is then possible to estimate the number of stream elements with the given minimum level $s$. Scaling that number by $2^s$ gives an approximation of the total number of distinct stream elements.
As mentioned in the introduction, the algorithm by Bar-Yossef can be extended with a merge operation, allowing it to be run in the distributed streams model. This essentially works by taking the maximum of the level stored in each bin and relies on the fact that $\max$ is a commutative, associative operation.

Kane et al.\ in 2010 found a solution to avoid the $\ln \ln n$ factor in the term $\eps^{-2} \ln \ln n$ of the space complexity, i.e., they were able to store a constant number of bits on average per bin instead of $\ln \ln n$ bits. To achieve this, instead of estimating the sub-sampling threshold at the end, they obtain a rough estimate for the cardinality of the set during the course of the algorithm. Whenever the rough estimate indicates a large enough sub-sampling threshold, the information in bins with smaller levels is not going to be needed. (Note that the estimate determined by the rough-estimation algorithm is monotone.) Besides dropping the data in bins with a maximal level below the current sub-sampling threshold, which I will refer to as the cut-level in the following, they only store the difference between the level of the element in each bin and the cut-level.
It is then possible to show that the expected number of bits necessary to store the compressed table values is on average $\bigo(1)$. To limit the space usage unconditionally the algorithm keeps track of the space usage for the table and, if it exceeds a constant times the table size, the algorithm will reach an error state deleting all information. To succeed they estimate the probability that the rough estimate is correct at all points during the course of the algorithm. Similarly they show that the space usage will be at most a constant times the bin count with high probability assuming the latter is true. To achieve that they rely on the fact that there are at most $\bigo(\ln n)$ points, where the rough estimate increases, which enables a union bound to verify that the error state will not be reached at or at any point before the estimation step.

However the bound on the number of changes in the rough estimatator is only true, when the algorithm is executed sequentially and their analysis does not extend to the distributed streams model.

This is a point where the solution presented here distinguishes itself: The algorithm in this work does not use a rough estimation algorithm to determine the cut-level. Instead, a cut-level is initialized to $0$ at the beginning and is increased if and only if the space usage would be too high otherwise. The algorithm never enters a failure state, preserving as much information as possible in the available memory. Because of the monotonicity of the values in the bins it is possible to show that the state of the algorithm is history independent. In the estimation step a sub-sampling threshold is determined using the values in the bins directly. This is distinct from the previously know methods, where two distinct data structures are being maintained in parallel. During the analysis it is necessary to take into account that the threshold is not independent of the values in the bins, which requires a slightly modified proof (see Lemma~\ref{le:e_2}). The proof that the cut-level will not be above the sub-sampling threshold works by verifying that the cut-level (resp. sub-sampling threshold) will with high probability be below (resp. above) a certain threshold that is chosen in the proof depending on the cardinality of the set (see Subsection~\ref{sec:overall}).

Another crucial idea introduced by Kane et al. is the use of a two-stage hash function, when mapping the universe elements from $[n]$ to $[b]$. The first hash function is selected from a pairwise hash-family mapping from $[n]$ to $[\tilde b]$ and the second is a $k$-wise independent family from $[\tilde b]$ to $[b]$. The value $\tilde b$ is chosen such that w.h.p. there are no collisions during the application of the first hash function (for the universe elements above the sub-sampling threshold), in that case the two-stage hash function behaves like a single-stage $k$-wise independent hash function from $[n]$ to $[b]$. This is achievable with a choice of $\tilde b \in \bigo(b^2)$ thus requiring fewer random bits than a single stage function from a $k$-wise independent family would.

The algorithm by Kane et al.\ discussed before this paragraph has a space complexity of $\bigo(\eps^{-2} + \ln n)$ for a fixed failure probability ($<\frac{1}{2}$). It is well known that the success probability of such an algorithm can be improved by running $l \in \bigo( \ln(\delta^{-1}) )$ independent copies of the algorithm and taking the median of the estimates of each independent run.~\cite{alon1999}[Thm. 2.1] In summary, this solution, as pointed out by Kane et al., leads to a space complexity of $\bigo( \ln(\delta^{-1}) (\eps^{-2} + \ln n))$. B\l{}asiok observed that this can be further improved: His main technique is to choose seed values of the hash functions using a random walk of length $l$ in an expander graph instead of independently. This reduces the space complexity for the seed values of the hash functions. Similarly he introduces a delta compression scheme for the states of the rough-estimation algorithms [which also have to be duplicated $l$ times]. In a straightforward manner his solution works only for the case where $\eps < (\ln n)^{-1/4}$. In the general case, he needs a more complex pseudo-random construction building on expander walks, Parvaresh-Vardy codes and a sub-sampling step. The main obstacle is the fact that deviation bounds for unbounded functions sampled by a random walk do not exist, even with doubly-exponential tail bounds.

In this work, because there is no distinct rough estimation data structure, the compression of its state is not an issue. However there is still the space usage for the cut-levels: Because maintaining a cut-level for each copy would require too much space, it is necessary to share the cut-level between at least $\ln \ln n$ copies at a time. 
To achieve that I use a two-stage expander construction. This means that each vertex of the first stage expander encodes a walk in a second expander. (Here it is essential that the second expander is regular.) The length of the walk of the second expander is $\bigo(\ln \ln n)$ matching the number of bits required to store a cut-level, while the length of the walk in the first (outer) expander is $\bigo( \ln (\delta^{-1}) (\ln \ln n)^{-1})$. Note that the product is again just $\bigo(\ln(\delta^{-1}))$. The key difference is that the copies in the inner expander have to share the same cut-level, while the outer walk does not, i.e.\ there are $\bigo( \ln (\delta^{-1}) (\ln \ln n)^{-1})$ separate cut-levels. See also Figure~\ref{fig:outer_inner_state}. To work this out the spectral gaps have to be chosen correctly and I introduce a new deviation bound for expander walks (in Section~\ref{sec:chernoff}.) This relies on a result mentioned in Impagliazzo and Kabanets from 2010~\cite{impagliazzo2010}, which shows a Chernoff bound for expander walks in terms of Kullback-Leibler divergence. Before we can detail that out let us first briefly introduce notation.

\section{Notation and Preliminaries}
This section summarizes (mostly standard) notation and concepts used in this work:
General constants are indicated as $C_1, C_2, \cdots$ etc. Their values are fixed throughout this work and are summarized in Table~\ref{tab:consts}.
For $n \in \mathbb N$, let us define $[n] := \{0, 1, \ldots, n-1\}$. The notation $[P]$ for a predicate $P$ denotes the Iverson bracket, i.e., its value is $1$ if the predicate is true and $0$ otherwise.
The notation $\ld x$ (resp. $\ln x$) stands for the logarithm to base $2$ (resp.\ $e$) of $x \in \mathbb R_{> 0}$. The notations $\floor{x}$ and $\ceil{x}$ represent the floor and ceiling functions: $\mathbb R \rightarrow \mathbb Z$.
For a probability space $\Omega$, the notation $\prob_{\omega \sim \Omega}(F(\omega))$ is the probability of the event: $\{\omega | F(\omega)\}$. And $\expect_{\omega \sim \Omega}(f(\omega))$ is the expectation of $f$ if $\omega$ is sampled from the distribution $\Omega$, i.e., $\expect_{\omega \sim \Omega}(f(\omega)) := \int_\Omega f(\omega) \, d \omega$. Similarly, $\var f = \expect (f - \expect f)^2$.
For a finite non-empty set $S$, $U(S)$ is the uniform probability space over $S$, i.e., $\prob(\{x\}) = \size{S}^{-1}$ for all $x \in S$. (Usually, we will abbreviate $U(S)$ with $S$ when it is obvious from the context.) All probability spaces mentioned in this work will be discrete, i.e., measurability will be trivial.

All graphs in this work are finite and are allowed to contain parallel edges and self-loops. For an ordering of the vertices of such a graph, it is possible to associate an adjacency matrix $A = (a_{ij})$, where $a_{ij}$ is the count of the edges between the $i$-th to the $j$-th vertex. We will say it is undirected $d$-regular if the adjacency matrix is symmetric and all its row (or equivalently) column sums are $d$. Such an undirected $d$-regular graph is called a $\lambda$-expander if the second largest absolute eigenvalue of its adjacency matrix is at most $d \lambda$.

Given an expander graph $G$, we denote by $\mathrm{Walk}(G,l)$, the set of walks of length $l$. For a walk $w \in \mathrm{Walk}(G,l)$ we write $w_i$ for the $i$-th vertex and $w_{i,i+1}$ for the edge between the $i$-th and $(i+1)$-th vertex. Because of the presence of parallel edges, two distinct walks may have the same vertex sequence. As a probability space $U(\mathrm{Walk(G,l)})$ corresponds to choosing a random starting vertex and performing an $(l-1)$-step random walk.

\section{Chernoff-type estimates for Expander Walks\label{sec:chernoff}}
The following theorem has been shown implicitly by Impagliazzo and Kabanets~\cite[Th.~10]{impagliazzo2010}:%
\begin{theorem}[Impagliazzo and Kabanets]\label{th:expander_chernoff}
Let $G = (V,E)$ be a $\lambda$-expander graph and $f$ a boolean function on its vertices, i.e.: $f : V \rightarrow \{0,1\}$ s.t. 
$\mu = \expect_{v \sim U(V)} f(v)$, $6 \lambda \leq \mu$ and $2 \lambda < \eps < 1$ then:
\[
    \prob_{w \sim \mathrm{Walk}(G,l)} \left( \textstyle \sum_{i \in [l]} f(w_i) \geq (\mu + \eps) l \right) \leq \exp( - l D( \mu + \eps || \mu + 2 \lambda ) )
\]
\end{theorem}
Especially, the restriction $\mu \geq 6 \lambda$ in the above result causes technical issues since usually one only has an upper bound for $\mu$. The result follows in Impagliazzo and Kabanets work as a corollary from the application of their main theorem~\cite{impagliazzo2010}[Thm.\ 1] to the hitting property established by Alon et al.~\cite[Th. 4.2]{alon1995} in 1995. It is easy to improve Theorem~\ref{th:expander_chernoff} by using an improved hitting property:
\begin{theorem}[Hitting Property for Expander Walks]
\label{th:expander_walk_hitting}
Let $G = (V,E)$ be a $\lambda$-expander graph and $W \subseteq V$, $I \subseteq [l]$ and let $\mu := \frac{\size{W}}{\size{V}}$ then:
\[
    \prob_{w \sim \mathrm{Walk}(G,l)} \left( \textstyle \bigwedge_{i \in I} w_i \in W \right) \leq (\mu (1-\lambda)  + \lambda)^{\size{I}} \leq (\mu+\lambda)^{\size{I}}
\]
\end{theorem}

\begin{proof}
The above theorem for the case where $I = [l]$ is shown by Vadhan~\cite[Theorem 4.17]{vadhan2012}. It is however possible to extend the proof to the case where $I \subset [l]$.
To understand that, it is important to note that the proof establishes that the wanted probability is the $l_1$ norm of $r := P (A P)^{l-1} u$ where $A$ is the transition matrix of the graph, $P$ is a diagonal matrix whose diagonal entries are in $\{0,1\}$ depending on whether the vertex is in the set $W$ and $u$ is the vector, where each component is $|V|^{-1}$. Note that $u$ represents the stationary distribution of the random walk. If $I$ is a strict subset of $[l]$ then the above term for $r$ needs to be corrected, by removing multiplications by $P$ for the corresponding steps, i.e.:
\[
    r' := A^{k_0} P A^{k_1} P A^{k_2} \ldots A^{k_{\size{I}-1}} P A^{k_{\size{I}}} u   
\]
where $k_i$ is distance between the $i$-th index in $I$ and $i+1$-th index in $I$.\footnote{The $0$-th index in $I$ is defined to be $0$ and the $\size{I}+1$-th index in $I$ is defined to be $l$.} Because $A u = u$ and because the application of $A$ does not increase the $l_1$ norm, it is possible to ignore the first and last term, i.e., it is enough to bound the $l_1$ norm of
\[
    x = P A^{k_1} P A^{k_2} \ldots A^{k_{\size{I}-1}} P u \textrm{.}
\]
This can be regarded as an $\size{I}$-step random walk, where the transition matrix is $A^{k_i}$ for step $i$. (Note that $k_i > 0$ for $1 \leq i \leq \size{I}-1$). The proof of the mentioned theorem~\cite[Thm.\ 4.17]{vadhan2012} still works in this setting if we take into account that $A^k$ is itself the adjacency matrix of a $\lambda$-expander on the same set of vertices. (Indeed it is even a $\lambda^k$-expander.)
\end{proof}%
With the previous result, it is possible to obtain a new, improved version of Theorem~\ref{th:expander_chernoff}:%
\begin{theorem}[Improved version of Theorem~\ref{th:expander_chernoff}]\label{th:expander_chernoff_improved}
Let $G = (V,E)$ be a $\lambda$-expander graph and $f$ a boolean function on its vertices, i.e.: $f : V \rightarrow \{0,1\}$ s.t. 
$\mu = \expect_{v \sim U(V)} f(v)$ and $\mu + \lambda \leq \gamma \leq 1$ then:
\[
    \prob_{w \sim \mathrm{Walk}(G,l)} \left( \textstyle \sum_{i \in [l]} f(w_i) \geq \gamma l \right) \leq \exp \left( - l D\left( \gamma || \mu + \lambda \right) \right)
\]
\end{theorem}
\begin{proof}
This follows from Theorem~\ref{th:expander_walk_hitting} and the generalized Chernoff bound \cite{impagliazzo2010}[Thm.\ 1].
\end{proof}
Impagliazzo and Kabanets approximate the divergence $D( \gamma || \mu + \lambda )$ by $2(\gamma - (\mu+\lambda))^2$. In this work, we are interested in the case where $\mu + \lambda \rightarrow 0$, where such an approximation is too weak, so we cannot follow that approach. (Note that $D(\gamma || \mu + \lambda)$ can be arbitrarily large, while $(\gamma - (\mu + \lambda))^2$ is at most $1$.) Instead, we derive a bound of the following form:
\begin{lemma}
\label{le:expander_chernoff}
Let $G = (V,E)$ be a $\lambda$-expander graph and $f$ a boolean function on its vertices, i.e.: $f : V \rightarrow \{0,1\}$ s.t. 
$\mu = \expect_{v \sim U(V)} f(v)$ and $\mu + \lambda \leq \gamma < 1$ then:
\[
    \prob_{w \sim \mathrm{Walk}(G,l)} \left( \textstyle \sum_{i \in [l]} f(w_i) \geq \gamma l \right) \leq \exp \left( - l (\gamma \ln ((\mu + \lambda)^{-1}) - 2e^{-1}) \right)
\]
\end{lemma}
\begin{proof}
The result follows from Theorem~\ref{th:expander_chernoff_improved} and the inequality:
$D( \gamma || p ) \geq \gamma \ln(p^{-1}) - 1$ for $0 < \gamma < 1$ and $0 < p < 1$.%

To verify that note:
\begin{eqnarray*}
    D( \gamma || p ) & \geq & \gamma \ln \gamma + \gamma \ln(p^{-1}) + (1-\gamma) \ln (1-\gamma) + (1-\gamma) \ln((1-p)^{-1}) \\
    & \geq & -e^{-1} + \gamma \ln (p^{-1}) - e^{-1} + 0 \geq \gamma \ln (p^{-1}) - 2e^{-1}
\end{eqnarray*}
using $x \ln x \geq -e^{-1}$ for $x > 0$ (and $\ln y \geq 0$ if $y \geq 1$).
\end{proof}
An application for the above inequality, where the classic Chernoff-bound by Gillman~\cite{gillman1998} would not be useful, is establishing a failure probability for the repetition of an algorithm that already has a small failure probability. For example, if an algorithm has a failure probability of $\delta^*$, then it is possible to repeat it $\bigo\left(\frac{\ln ( \delta^{-1})}{\ln ((\delta^*)^{-1})}\right)$-times to achieve a failure probability of $\delta$.  (This is done in Section~\ref{sec:ext_arb}.) Another consequence of this is a deviation bound for unbounded functions with a sub-gaussian tail bound:
\begin{restatable}[Deviation Bound]{lemma}{deviationboundstatement}
\label{le:deviation_bound}
Let $G = (V,E)$ be a $\lambda$-expander graph and $f : V \rightarrow \mathbb R_{\geq 0}$ s.t. $\prob_{v \sim U(V)}( f(v) \geq x) \leq \exp(-x (\ln x)^3)$ for $x \geq 20$ and $\lambda \leq \exp(-l (\ln l)^3 )$ then
\[
    \prob_{w \sim \mathrm{Walk}(G,l)} \left( \textstyle \sum_{i \in [l]} f(w_i) \geq \const{c:dev_bound} l \right) \leq \exp(-l)
\]
where $\const{c:dev_bound} := e^2 + e^3 + (e-1) \leq 30$.
\end{restatable}
Note that the class includes sub-gaussian random variables but is even larger. The complete proof is in Appendix~\ref{apx:proof_dev_bound}\@. The proof essentially works by approximating the function $f$ using the Iverson bracket:
$f(x) \leq \Sigma_k [e^k \leq f(x) \leq e^{k+1}] e^{k+1}$ and establishing bounds on the frequency of each bracket. For large $k$ this is established using the Markov inequality, and for small $k$ the previous lemma is used. The result is a stronger version of a lemma established by B\l{}asiok~\cite{blasiok2020}[Lem.~36], and the proof in this work is heavily inspired by his.\footnote{The main distinction is that he relies on a tail bound from Rao~\cite{rao2017}, while this work relies on Lemma~\ref{le:expander_chernoff}.}
\section{Explicit Pseudo-random Constructions}
This section introduces two families of pseudo-random objects used in this work along with an explicit construction for each. %(For both objects there are alternatives, such as tabulation-based hashing or different expander families.)
\subsection{Strongly explicit expander graphs}
For the application in this work, it is necessary to use \emph{strongly explicit expander graphs}. For such a graph, it is possible to sample a random vertex uniformly and compute the edges incident to a given vertex algorithmically, i.e., it is possible to sample a random walk without having to represent the graph in memory. Moreover, sampling a random walk from a $d$-regular graph $G$ with $n$-vertices is possible using a random sample from $[n d^{l-1}]$, i.e., we can map such a number to a walk algorithmically, such that the resulting distribution corresponds to the distribution from $\mathrm{Walk}(G,l)$ --- this allows the previously mentioned two-stage construction.

A possible construction for strongly explicit expander graphs for every vertex count $n$ and spectral bound $\lambda$ is described by Murtagh et al.~\cite{murtagh2019}[Thm. 20, Apx. B]\footnote{Similar results have also been discussed by Goldreich and Alon: Goldreich~\cite{goldreich2020} discusses the same problem but for edge expansion instead of the spectral bound. Alon~\cite{alon2020} constructs near-optimal expander graphs for every size starting from a minimum vertex counts (depending on the degree and discrepancy from optimality).}. Note that the degree $d$ in their construction only grows polynomially with $\lambda^{-1}$, hence $\ln (d(\lambda)) \in \bigo( \ln (\lambda^{-1}))$. We will use the notation $\mathcal E([n], \lambda, l)$ for the sample space of random walks of length $l$ in the described graph over the vertex set $[n]$. The same construction can also be used on arbitrary finite vertex sets $S$, if it is straightforward to map $[\size{S}]$ to $S$ algorithmically. Thus we use the notation $\mathcal E(S, \lambda, l)$ for such $S$. Importantly $\size{\mathcal E(S, \lambda, l)} = \size{S} d(\lambda)^{l-1}$. Thus a walk in such a graph requires $\bigo( \ld \size{S} + l \ld (\lambda^{-1}))$ bits to represent. %As dicussed in the previous paragraph, it is possible to choose a walk from the probability space $\mathcal E(S, \lambda, l)$ by choosing a natural number from $\left[\size{\mathcal E(S, \lambda, l)}\right]$ uniformly.

\subsection{Hash Families}

Let us introduce the notation: $\mathcal H_k([2^d], [2^d])$ for the Carter-Wegman hash-family~\cite{wegman1981} from $[2^d]$ to $[2^d]$. (These consist of polynomials of degree less than $k$ over the finite field $\mathrm{GF}(2^d)$).
It is straightforward to see that a hash-family for a domain $[2^d]$ is also a family for a subset of the domain $[n] \subseteq [2^d]$. Similarly it is possible to reduce the size of the range by composing the hash function with a modulo operation: $[2^d] \rightarrow [2^c]$ for $c \leq d$. Hence the previous definition can be extended to hash families with more general domains and ranges, for which we will use the notation: $\mathcal H_k([n],[2^c])$.
Note that $\ld \left(\size{\mathcal H_k([n],[2^c])}\right) \in \bigo( k (c + \ln n) )$.

For our application, we will need a second family with a geometric distribution (as opposed to uniform) on the range, in particular such that $\prob (f(a) \geq k) = 2^{-k}$. This is being used to assign levels to the stream elements. A straightforward method to achieve that is to compose the functions of the hash family $\mathcal H_k([2^d],[2^d])$ with the function that computes the number of trailing zeros of the binary representation of its input $[2^d] \rightarrow [d]$. We denote such a hash family with $\mathcal G_k([2^d])$ where the range is $[d+1]$. Like above, such a hash family is also one for a domain $[n] \subseteq [2^d]$, and hence we can again extend the notation: $\mathcal G_k([n])$.
Note that: $\prob_{f \sim G_k([n])} (f(a) \geq k) = 2^{-k} \textrm{ for all } k \leq \ceil{\ld n}$ and also $\ld \left(\size{\mathcal G_k([n])}\right) \in \bigo( k \ln n)$.

\section{The Algorithm\label{sec:alg}}
Because of all the distinct possible execution models, it is best to present the algorithm as a purely functional data structure with four operations:
\begin{align*}
\mathrm{init}:&\, () \rightarrow \mathrm{seed} &
\mathrm{single}:&\, [n] \rightarrow \mathrm{seed} \rightarrow \mathrm{sketch} \\
\mathrm{merge}:&\, \mathrm{sketch} \rightarrow \mathrm{sketch} \rightarrow \mathrm{sketch} &
\mathrm{estimate}:&\, \mathrm{sketch} \rightarrow \mathbb R
\end{align*}
The $\mathrm{init}$ step should be called only once globally --- it is the only random operation --- its result forms the seed and must be the same during the entire course of the algorithm. The operation $\mathrm{single}$ returns a sketch for a singleton set corresponding to its first argument. The operation $\mathrm{merge}$ computes a sketch representing the union of its input sketches and the operation $\mathrm{estimate}$ returns an estimate for the number of distinct elements for a given sketch. It is possible to introduce another primitive for adding a single element to a sketch, which is equivalent to a $\mathrm{merge}$ and a $\mathrm{single}$ operation, i.e.: $\mathrm{add}(x,\tau,\omega) := \mathrm{merge}(\tau, \mathrm{single}(x, \omega))$. In terms of run-time performance it makes sense to introduce such an operation, especially with an in-place update, but we will not discuss it here. %(It is however straightforward.)

The algorithm will be introduced in two successive steps. The first step is a solution that works for $(\ln n)^{-1} \leq \delta < 1$. The sketch requires only $\bigo(\ln(\delta^{-1}) \eps^{-2} + \ln \ln n)$, but the initial coin flips require $\bigo(\ln n + \ln (\eps^{-1})^2+\ln(\delta^{-1})^3)$ bits. For $\delta \geq (\ln n)^{-1}$ this is already optimal. In the second step (Section~\ref{sec:ext_arb}) a black-box vectorization of the previous algorithm will be needed to achieve the optimal $\bigo(\ln(\delta^{-1}) \eps^{-2} + \ln n)$ space usage for all $0 < \delta < 1$. %(Roughly this last step is a trade, where we increase the space usage for the mutable state to reduce the required random coin flips.)

For this entire section let us fix a universe size $n > 0$, a relative accuracy $0 < \eps < 1$, a failure probability $(\ln n)^{-1} \leq \delta < 1$ and define:
\begin{align*}
 l & := \ceil{\const{c:eps} \ln (2 \delta^{-1})} & b & := 2^{\ceil{\ld (\const{c:delta} \eps^{-2})}} \\
 k & := \ceil{\const{c:approx_bin_balls_1} \ln b + \const{c:approx_bin_balls_2}} & \lambda & := \min \left(\frac{1}{16}, \exp(-l (\ln l)^3)\right) \\
 \Psi & := \mathcal G_2([n]) \times \mathcal H_2([n],[\const{c:pre_bins} b^2]) \times \mathcal H_k([\const{c:pre_bins} b^2], [b]) & \Omega & := \mathcal E(\Psi,\lambda, l)
\end{align*}%
\begin{algorithm}[h]%
\begin{pseudo*}[indent-mark]
\kw{function} $\mathrm{init}()$ : $\Omega$ \\+
  \kw{return} $\mathrm{random} \, U(\Omega)$ \\-
\\
\kw{function} $\mathrm{compress}((B,q) : \mathcal S)$ : $\mathcal S$ \\+
  \kw{while} $\sum_{i \in [l],j \in [b]} \floor{\ld (B[i,j]+2)} > \const{c:space_bound} b l$ \\+
    $q \gets q + 1$ \\
    $B[i,j] \gets \max(B[i,j]-1,-1)$ \kw{for} $i \in [l], j \in [b]$ \\-
  \kw{return} $(B,q)$ \\-
\\
\kw{function} $\mathrm{single}(x : U, \omega : \Omega)$ : $\mathcal S$ \\+
  $B[i,j] \gets -1$ \\
  \kw{for} $i \in [l]$ \\+
    $B[i, h(g(x))] = f(x)$ where $(f, g, h) = \omega_i$ \\-
  \kw{return} $\mathrm{compress}(B,0)$ \\-
\\
\kw{function} $\mathrm{merge}((B_a,q_a) : \mathcal S, (B_b,q_b) : \mathcal S)$ : $\mathcal S$ \\+
  $q \gets \max (q_a,q_b)$ \\
  $B[i,j] \gets \max (B_a[i,j]+q_a-q, B_b[i,j]+q_b-q)$ \kw{for} $i \in [l], j \in [b]$ \\
  \kw{return} $\mathrm{compress}(B,q)$ \\-
\\
\kw{function} $\mathrm{estimate}((B,q) : \mathcal S)$ : $\mathbb R$ \\+
  \kw{for} $i \in [l]$ \\+
    $s \gets \max(0, \max \{ B[i,j]+q \mid | j \in [b] \} - \ld b + 9)$ \\
    $p \gets \size{\{ j \in [b] | B[i,j] + q \geq s\}}$ \\
    $Y_i \gets 2^{s} \ln (1-p b^{-1}) (\ln (1-b^{-1}))^{-1}$ \\-
  \kw{return} $\mathrm{median} (Y_0,\ldots, Y_{l-1})$
\end{pseudo*}
\caption{Algorithm for $\delta > (\ln n)^{-1}$} 
\label{alg:main}
\end{algorithm}%
The implementation of the operations is presented in Algorithm~\ref{alg:main}. Note that these are functional programs and pass the state as arguments and results; there is no global (mutable) state. The sketch consists of two parts $(B,q)$. The first part is a two-dimensional table of sizes $b$ and $l$. The second part is a single natural number, the cut-off level. The function $\mathrm{compress}$ is an internal operation and is not part of the public API. It increases the cut-off level and decreases the table values if the space usage is too high.

\subsection{History-Independence}
As mentioned in the introduction, this algorithm is history-independent, meaning that given the initial coin flips, it will reach the same state no matter in which permutation or frequency the stream elements are encountered. More precisely, the final state only depends on the set of encountered distinct elements over the execution tree and the initial coin flips, but not the shape of the tree. This is one of the key improvements compared to the solutions by Kane et al.\ and B\l{}asiok. Informally, this is easy to see because the chosen cut-off level is the smallest possible with respect to the size of the values in the bins, and that property is maintained because the values in the bins are monotonically increasing with respect to the set of elements in the execution tree. Nevertheless, let us prove the property more rigorously:

Let $\omega \in \Omega$ be the initial coin flips. Then there is a function $\tau(\omega, A)$ such that following equations hold:
\begin{eqnarray}
    \mathrm{single}(\omega,x) & = & \tau(\omega,\{x\}) \label{eq:single_hi} \\
    \mathrm{merge}(\tau(\omega,A),\tau(\omega,B)) & = & \tau(\omega,A \cup B) \label{eq:merge_hi}
\end{eqnarray}
The function $\tau$ is defined as follows:
\begin{align*}
    \tau_0((f,g,h),A) & := j \rightarrow \max \{ f(a) \mid a \in A \wedge h(g(a)) = j \} \cup \{-1\} \\
    \tau_1(\psi,A,q) & := j \rightarrow \max \{ \tau_0(\psi,A) - q, -1 \} \\
    \tau_2(\omega,A,q) & := (i,j) \rightarrow  \tau_1(\omega_i,A,q)[j] \\ 
    q(\omega,A) & := \min \left\{ q \geq 0 \, \middle | \, \textstyle \sum_{i \in [l], j \in [b]} \floor{ \ld( \tau_2(\omega,A,q)[i,j] + 2 )} \leq \const{c:space_bound} b l \right\} \\
    \tau_3(\omega,A,q) & := (\tau_2(\omega,A,q), q) \\ 
    \tau(\omega,A) & := \tau_3(\omega,A,q(\omega,A))
\end{align*}%
The function $\tau_0$ describes the values in the bins if there were no compression, i.e., when $q=0$. The function $\tau_1$ describes the same for the given cut-off level $q$. Both are with respect to the selected hash functions $\psi = (f,g,h)$.
The function $\tau_2$ represents the state of all tables based on a seed for the expander. The next function $\tau_3$ represents the entire state, which consists of the tables and the cut-off level. The function $q$ represents the actual cut-off level that the algorithm would choose based on the values in the bins. Finally, the full state is described by the function $\tau$ for a given seed $\omega$ and set of elements $A$.
\begin{lemma} \label{le:histind} Equations~\ref{eq:single_hi} and \ref{eq:merge_hi} hold for all $\omega \in \Omega$ and $\emptyset \neq A \subset [n]$.
\end{lemma}
\begin{proof}
Let us also introduce the algorithms $\mathrm{merge}_1$ and $\mathrm{single}_1$. These are the algorithms $\mathrm{merge}$ and $\mathrm{single}$ but without the final compression step. By definition, we have $\mathrm{merge}(x,y) = \mathrm{compress}(\mathrm{merge}_1(x,y))$ and, similarly, $\mathrm{single}(\omega,x) = \mathrm{compress}(\mathrm{single}_1(\omega,x))$.

The following properties follow elementarily\footnote{The verification relies on the semi-lattice properties of the $\max$ operator, as well as its translation invariance (i.e. $\max (a+c,b+c) = \max (a,b) + c$).} from the definition of $\tau$, $s$ and the algorithms:
\begin{enumerate}[(i)]
\item \label{e:tau_i} $\tau(\omega, A) = \mathrm{compress}(\tau_3(\omega,A,q))$ for all $0 \leq q \leq q(\omega,A)$
\item \label{e:tau_ii} $\tau_3(\omega, A_1 \cup A_2, \max(q(\omega,A_1),q(\omega,A_2))) = \mathrm{merge}_1(\tau(\omega, A_1), \tau(\omega, A_2))$
\item \label{e:tau_iii} $\tau_3(\{x\}, 0) = \mathrm{single}_1(\omega, x)$
\item \label{e:tau_iv} $q(\omega, A_1) \leq q(\omega, A_2)$ if $A_1 \subseteq A_2$
\item \label{e:tau_v} $q(A) \geq 0$
\end{enumerate}
To verify Eq.\ \ref{eq:single_hi} we can use \ref{e:tau_i}, \ref{e:tau_iii} and \ref{e:tau_v} and to verify Eq.\ \ref{eq:merge_hi} we use \ref{e:tau_i}, \ref{e:tau_ii} taking into account that 
$\max(q(\omega, A_1),q(\omega, A_2)) \leq q(\omega, A_1 \cup A_2)$ because of \ref{e:tau_iv}.
\end{proof}

\subsection{Overall Proof\label{sec:overall}}
Because of the argument in the previous section, $\tau(\omega,A)$ will be the state reached after any execution tree over the set $A$ and the initial coin flips, i.e., $\omega \in \Omega$. Hence for the correctness of the algorithm, we only need to show that:
\begin{theorem}\label{th:overall}
Let $\emptyset \neq A \subseteq [n]$ then
$\prob_{\omega \in U(\Omega)}\left( \mathrm{estimate}(\tau(\omega,A)) - \size{A} \geq \eps \size{A}\right) < \delta$.
\end{theorem}
Proof: Postponed. This will be shown in two steps: First, we want to establish that the cut-off threshold $q$ will be equal to or smaller than $q_\mathrm{max} := \max(0, \ceil{\ld \size{A}} - \ld b)$ with high probability. And if the latter is true, then the estimate will be within the desired accuracy with high probability. For the second part, we verify that the estimation step will succeed with high probability for all $0 \leq q \leq q_\mathrm{max}$. (This will be because the sub-sampling threshold $s$ in the estimation step will be $\geq q_\mathrm{max}$ with high probability.)

For the remainder of this section, let $\emptyset \neq A \subset [n]$ be fixed and we will usually omit the dependency on $A$. For example, we will write $\tau(\omega)$ instead of $\tau(\omega,A)$.

Formally we can express the decomposition discussed above using the following chain:
\begin{equation}
\label{eq:overall_chain}
\begin{aligned}
& \prob_{\omega \in \Omega}\left( \size{\mathrm{estimate}(\tau(\omega)) - \size{A}} \geq \eps \size{A} \right) \leq \\ 
& \prob_{\omega \in \Omega}\left( \exists q \leq q_\mathrm{max}. \size{\mathrm{estimate}(\tau_2(\omega,q)) - \size{A}} \geq \eps \size{A} \vee q(\omega) > q_\mathrm{max} \right) \leq \\
& \prob_{\omega \in \Omega}\left( \exists q \leq q_\mathrm{max}. \size{\mathrm{estimate}(\tau_2(\omega,q)) - \size{A}} \geq \eps \size{A}\right) + \prob_{\omega \in \Omega}\left( q(\omega) > q_\mathrm{max} \right) \leq \frac{\delta}{2} + \frac{\delta}{2}
\end{aligned}
\end{equation}

The first inequality is the converse of the informal argument from above.\footnote{Algebraically it is more succinct to bound the failure event from above, instead of bounding the success event from below, which means that some informal arguments will be accompanied by their algebraic converse. For example, an argument that event $A$ implies $B$ might be accompanied by $P(\neg B) \leq \prob(\neg A)$.} The second inequality is just the sub-additivity of probabilities. And the third inequality consists of the two goals we have, i.e., the overall proof can be split into two parts:

\begin{itemize}
\item $\prob_{\omega \in \Omega}\left( q(\omega) > q_\mathrm{max} \right) \leq \frac{\delta}{2}$
\item $\prob_{\omega \in \Omega}\left( \exists q \leq q_\mathrm{max}. \size{\mathrm{estimate}(\tau_2(\omega,q)) - \size{A}} \geq \eps \size{A}\right) \leq \frac{\delta}{2}$
\end{itemize}
The first will be shown in the following subsection, and the next in the subsequent one.
Subsection~\ref{sec:space_usage} discusses the space usage of the algorithm.
\subsection{Cut-off Level\label{sec:cut_level}}
This subsection proves that the cut-off level will be smaller than or equal to $q_\mathrm{max}$. This is the part where the tail estimate for sub-gaussian random variables over expander walks (Lemma~\ref{le:deviation_bound}) is applied:
\begin{lemma} \label{le:cut_level}
    $\prob_{\omega \in \Omega}\left( q(\omega) > q_\mathrm{max} \right) \leq \frac{\delta}{2}$
\end{lemma}
\begin{proof}
Let us make a few preliminary observations:
\begin{equation}
    \label{eq:cut_level_1}
    \floor{\ld(x+2)} \leq \ld(x+2) \leq (c+2)+ \max(x-2^c,0) \textrm{ for } (-1) \leq x \in \mathbb R \textrm{ and } c \in \mathbb N \textrm{.}
\end{equation}
This can be verified using case distinction over $x \geq 2^c+2$.
\begin{equation}
    \label{eq:cut_level_2}
    \expect_{f \sim \mathcal G_2([n])} \max (f(a) - q_\mathrm{max} - 2^c, 0) \leq 2^{-q_\mathrm{max}} 2^{-2^c} \textrm{ for all } a \in [n] \textrm{ and } c \in \mathbb N
\end{equation}
Note that this relies on the fact $f$ is geometrically distributed.
\begin{equation}
    \label{eq:cut_level_3}
\size{A} b^{-1} 2^{-q_\mathrm{max}} \leq 1
\end{equation}
This follows from the definition of $q_\mathrm{max}$ via case distinction.

To establish the result, we should take into account that $q(\omega)$ is the smallest cut-off level $q$ fulfilling the inequality:
$\sum_{i \in [l], j \in [b]} \floor{\ld( \tau_2(\omega,q)[i,j] + 2 )} \leq \const{c:space_bound} b l$.
In particular, if the inequality is true for $q_\mathrm{max}$, then we can conclude that $q(\omega)$ is at most $q_\mathrm{max}$, i.e.:
\begin{equation}\label{eq:cut_level_main}
\prob_{\omega \in \Omega}\left( q(\omega) > q_\mathrm{max}\right) = \prob_{\omega \in \Omega} \left( \sum_{i \in [l], j \in [b]} \floor{\ld( \tau_2(\omega,q_\mathrm{max})[i,j] + 2 )} > \const{c:space_bound} b l \right)
\end{equation}

Let us introduce the random variable $X$ over the seed space $\Psi$. It describes the space usage of a single column of the table $B$:
\[
    X(\psi) := \sum_{j \in [b]} \floor{ \ld( \tau_1(\psi,q_\mathrm{max})[j] + 2)}
\]
Which can be approximated using Eq.~\ref{eq:cut_level_1} as follows:
\begin{equation*}
    X(\psi) \leq \sum_{j \in [b]} c+2 + \max ( \tau_1(\psi,q_\mathrm{max})[j] - 2^c, 0) = \sum_{j \in [b]} c+2 + \max (\tau_0(\psi)[j] - q_\mathrm{max} - 2^c, 0 )
\end{equation*}
for all $0 \leq c \in \mathbb N$. Hence:
\begin{align*}
    & \prob_{\psi \sim \Psi} \left( X(\psi) \geq (c+3) b \right) \leq \prob_{\psi \sim \Psi} \left( \textstyle \sum_{j \in [b]} \max (\tau_0(\psi)[j] - q_\mathrm{max} - 2^c, 0 ) \geq b \right) \leq \\
    & \prob_{(f,g,h) \sim \Psi} \left( \textstyle \sum_{j \in [b]} \max \{ f(a) - q_\mathrm{max} - 2^c \mid a \in A \wedge h(g(a)) = j \} \cup \{0\} \geq b \right) \leq \\
    & \prob_{(f,g,h) \sim \Psi} \left( \textstyle \sum_{a \in A} \max (f(a) - q_\mathrm{max} - 2^c, 0) \geq b \right) \leq \\
    & b^{-1} \sum_{a \in A} \expect_{(f,g,h) \sim \Psi} \max (f(a) - q_\mathrm{max} - 2^c, 0) \leq  
    b^{-1} \size{A} 2^{-q_\mathrm{max}} 2^{-2^c} \leq 2^{-2^c}
\end{align*}
where the third and second-last inequality follow from Eq.~\ref{eq:cut_level_2} and \ref{eq:cut_level_3}.
It is straightforward to conclude from the latter that for \emph{all} $20 \leq x \in \mathbb R$:
\begin{align*}
    \prob_{\psi \sim \Psi} \left( \frac{X(\psi)}{b} - 3 \geq x \right) & \leq \prob_{\psi \sim \Psi} \left( X(\psi) \geq b(\lfloor x \rfloor+3) \right) \leq \exp ( -2^{\lfloor x \rfloor} \ln 2) \leq e^{-x (\ln x)^3}
\end{align*}
Hence, it is possible to apply Lemma~\ref{le:deviation_bound} on the random variables $b^{-1} X(\psi) - 3$ obtaining:
\[
    \prob_{\omega \in \Omega} \left( \textstyle \sum_{i \in [l]} b^{-1} X(h(\omega,i)) - 3 \geq \const{c:dev_bound} l \right) \leq \exp(-l) \leq \frac{\delta}{2}
\] 
This lemma now follows using $\const{c:space_bound} \geq \const{c:dev_bound} + 3$ and that $\sum_{i \in [l]} X(h(\omega,i)) \leq \const{c:space_bound} b l$ implies $q(\omega) \leq q_\mathrm{max}$ as discussed at the beginning of the proof (Eq.~\ref{eq:cut_level_main}).
\end{proof}

\subsection{Accuracy}
Let us introduce the random variables:
\begin{align*}
    t(f) & := \max \{ f(a) \mid a \in A \} - \ld b + 9 &
    s(f) & := \max(0, t(f)) \\
    p(f,g,h) & := \size{\{ j \in [b] \mid \tau_1((f,g,h), 0)[j] \geq s(f) \}} &
    Y(f,g,h) & := 2^{s(f)} \rho^{-1}(p(f,g,h))
\end{align*}
where $\rho(x) := b (1-(1-b^{-1})^x)$ --- the expected number of hit bins when $x$ balls are thrown into $b$ bins. (See also Figure~\ref{fig:rho}). Note that the definitions $t$, $p$ and $Y$ correspond to the terms within the loop in the $\mathrm{estimate}$ function under the condition that the approximation threshold $q$ is $0$.
In particular: $\mathrm{estimate}(\tau_3(\omega,0)) = \mathrm{median}_{i \in [l]} Y(\omega_i)$ for $\omega \in \Omega$.
\begin{figure}[t]
\centering
\includegraphics{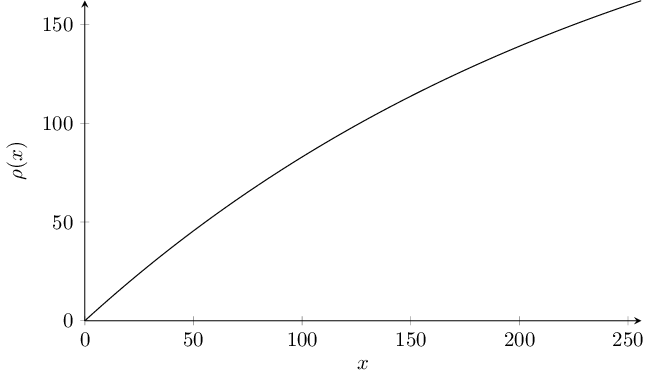}
\caption{Plot for $\rho(x)$ in the case $b=256$.}\label{fig:rho}
\end{figure}%
Moreover, we denote by $R(f)$ the set of elements in $A$ whose level is above the sub-sampling threshold, i.e.:
$R(f) := \{ a \in A \mid f(a) \geq s(f) \}$.
The objective is to show that the individual estimates obtained in the loop in the $\mathrm{estimate}$ function (assuming $q=0$) have the right accuracy and that the threshold $s \geq q_\mathrm{max}$ with high probability, i.e.:
\begin{equation}
\label{eq:median_single}
    \prob_{\psi \sim \Psi}\left( \size{Y(\psi) - \size{A}} > \eps \size{A} \vee s(f) < q_\mathrm{max} \right) \leq \frac{1}{16}
\end{equation}
In Lemma~\ref{le:median_single} this will be generalized to $0 \leq q \leq q_\mathrm{max}$.
To be able to establish a bound on the above event, we need to check the likelihood of the following $4$ events:
\begin{itemize}
\item The computed sub-sampling threshold $s(f)$ is approximately $\ld (\size{A})$.
\item The size of the sub-sampled elements $R(f)$ is a good approximation of $2^{-s(f)} \size{A}$.
\item There is no collision during the application of $g$ on the sub-sampled elements $R(f)$.
\item The count of elements above the sub-sampling threshold in the table is close to the expected number $\rho(R(f))$ (taking collisions due to the application of $h$ into account).
\end{itemize}
Then it will be possible to conclude that one of the above must fail if the approximation is incorrect. More formally:
\begin{align*}
    E_1(\psi) & :\leftrightarrow 2^{-16} b \leq 2^{-t(f)} \size{A} \leq 2^{-1} b &
    E_2(\psi) & :\leftrightarrow \size{\size{R(f)} - 2^{-s(f)} \size{A}} \leq \textstyle \frac{\eps}{3} 2^{-s(f)} \size{A} \\
    E_3(\psi) & :\leftrightarrow \forall a \neq b \in R(f) . g(a) \neq g(b) &
    E_4(\psi) & :\leftrightarrow \size{p(\psi) - \rho(\size{R(f)})} \leq \textstyle \frac{\eps}{12} \size{R(f)}
\end{align*}
for $\psi = (f,g,h) \in \Psi$. The goal is to show all four events happen simultaneously w.h.p.: 
\begin{equation}
\label{eq:median_single_pre}
    \prob_{\psi \sim \Psi} ( \neg E_1(\psi) \vee \neg E_2(\psi) \vee \neg E_3(\psi) \vee \neg E_4(\psi) ) \leq \frac{1}{16}
\end{equation}
A first idea might be to establish the above by showing separately that:
$\prob_{\psi \sim \Psi} ( \neg E_i(\psi) ) \leq 2^{-6}$ for each $i \in \{1,\ldots,4\}$. However this does not work and the actual strategy is to establish bounds on
$\prob_{\psi \sim \Psi} \left( \bigwedge_{j < i} E_j(\psi) \wedge \neg E_i(\psi) \right) \leq 2^{-6}$ for each $i \in \{1,\ldots,4\}$.
Note that the latter still implies Equation~\ref{eq:median_single_pre}. Let us start with the $i=1$ case:
\begin{lemma} \label{le:e_1}
$\prob_{\psi \in \Psi} ( \neg E_1(\psi) ) \leq 2^{-6}$
\end{lemma}
\begin{proof}
For $X(f) = \max \{ f(a) \mid a \in A \}$ it is possible to show:
\begin{align*}
    \prob_{(f,g,h) \sim \Psi} & \left( X(f) < \ld(\size{A}) - k - 1 \right) \leq 2^{-k} &
    \prob_{(f,g,h) \sim \Psi} & \left( X(f) > \ld(\size{A}) + k \right) \leq 2^{-k} 
\end{align*}
using the proof for the $F_0$ algorithm by Alon et al.~\cite{alon1999}[Proposition~2.3].
The desired result follows taking $k=7$ and that $t(f) = X(f) - \ld b + 9$.
\end{proof}

The following lemma is the interesting part of the proof in this subsection. In previous work, the sub-sampling threshold is obtained using a separate parallel algorithm, which has the benefit that it is straightforward to verify that $\size{R(f)}$ approximates $2^{-s} \size{A}$. The drawback is, of course, additional algorithmic complexity and an additional independent hash function. However, in the solution presented here, the threshold is determined from the data to be sub-sampled itself, which means it is not possible to assume independence. The solution to the problem is to show that $\size{R(f)}$ approximates $2^{-s} \size{A}$ with high probability for \emph{all} possible values $s(f)$ assuming $E_1$.

\begin{lemma} \label{le:e_2}
$L := \prob_{\psi \sim \Psi} ( E_1(\psi) \wedge \neg E_2(\psi) ) \leq 2^{-6}$
\end{lemma}
\begin{proof}
Let $r(f,t) := \size{ \{ a \in A \mid f(a) \geq t \} }$ and $t_\mathrm{max}$ be maximal, s.t.\ $2^{-16} b \leq 2^{-t_\mathrm{max}} \size{A}$.
Then $2^7 \leq \frac{\eps^2}{9} 2^{-16} b \leq \frac{\eps^2}{9} 2^{-t_\mathrm{max}} \size{A}$.
Hence: $2^{7+t_\mathrm{max}-t} \leq \frac{\eps^2}{9} 2^{-t} \size{A} = \frac{\eps^2}{9} \expect r(f,t)$.
Thus:
\[
    2^{7+t_\mathrm{max}-t} \var r(f,t) \leq 2^{7+t_\mathrm{max}-t} \expect r(f,t) \leq \frac{\eps^2}{9} (\expect r(f,t))^2
\]
for all $0 < t \leq t_\mathrm{max}$. (This may be a void statement if $t_\mathrm{max} \leq 0$.)
Hence:
\begin{align*}
    & \prob_{(f,g,h) \in \Psi}\left( \exists t. 0 < t \leq t_\mathrm{max} \wedge \size{r(f,t) - \expect r(\cdot,t)} > \frac{\eps}{3} \expect r(\cdot,t) \right) \leq \\
    & \sum_{t=1}^{t_\mathrm{max}} \prob_{(f,g,h) \in \Psi}\left( \size{r(f,t) - \expect r(\cdot,t)} > \sqrt{ 2^{7+t_\mathrm{max}-t} \var r(f,t)} \right) \leq 
     \sum_{t=1}^{t_\mathrm{max}} 2^{-7-t_\mathrm{max}+t} \leq 2^{-6}
\end{align*}
Note that the predicate $E_2(\psi)$ is always true if $s(f) = 0$ because, in that case, there is no sub-sampling, i.e., $\size{R(f)} = \size{A}$. On the other hand if $s(f) > 0$, then $s(f) = t(f) \leq t_\mathrm{max}$ assuming $E_1(\psi)$.
Hence:
\begin{eqnarray*}
    L & \leq & \prob_{(f,g,h)} \left( s(f) > 0 \wedge E_1(f,g,h) \wedge \neg E_2(f,g,h) \right) \\
    & \leq & \prob_{(f,g,h)} \left( 0 < t(f) \leq t_\mathrm{max} \wedge \size{\size{R(f)} - 2^{-t(f)} \size{A}} > \textstyle \frac{\eps}{3} 2^{-t(f)} \size{A} \right) \\
    & \leq & \prob_{(f,g,h)} \left( 0 < t(f) \leq t_\mathrm{max} \wedge \size{r(f,t(f)) - 2^{-t(f)} \size{A}} > \textstyle \frac{\eps}{3} 2^{-t(f)} \size{A} \right) \leq 2^{-6}
\end{eqnarray*}
where the last step follows from the previous equation.
\end{proof}%
\begin{equation}\label{eq:B_bound}
\textrm{Note that: }  E_1(f,g,h) \wedge E_2(f,g,h) \rightarrow \size{R(f)} \leq \frac{2}{3} b \textrm{ for } (f,g,h) \in \Psi
\end{equation}%
\begin{lemma} \label{le:e_3}
$L := \prob_{\psi \sim \Psi} ( E_1(\psi) \wedge E_2(\psi) \wedge \neg E_3(\psi) ) \leq 2^{-6}$
\end{lemma}
\begin{proof}
Using Eq.~\ref{eq:B_bound} we can conclude:
\begin{eqnarray*}
    L & \leq & \prob_{(f,g,h) \sim \Psi} \left( \size{R(f)} \leq b \wedge (\exists a < b \in R(f). g(a) = g(b)) \right) \\
     & \leq & \int_{\mathcal G_2([n])} [\size{R(f)} \leq b] \prob_{g \sim \mathcal H_2([n], [\const{c:pre_bins} b^2])}(\exists a < b \in R(f). g(a) = g(b)) \, d f \\
     & \leq & \int_{\mathcal G_2([n])} [\size{R(f)} \leq b] \sum_{a < b \in R(f)} \prob_{g \sim \mathcal H_2([n], [\const{c:pre_bins} b^2])}(g(a) = g(b)) \, d f \\
     & \leq & \int_{\mathcal G_2([n])} \frac{b (b-1)} {2 \const{c:pre_bins} b^2} \, d f \leq \frac{1}{2 \const{c:pre_bins}} = 2^{-6} \textrm{.\qedhere}
\end{eqnarray*}%
\end{proof}%
\begin{lemma} \label{le:e_4}
$L := \prob_{\psi \sim \Psi} ( E_1(\psi) \wedge E_2(\psi) \wedge E_3(\psi) \wedge \neg E_4(\psi) ) \leq 2^{-6}$
\end{lemma}
\begin{proof}
Let $\tilde{R}(f,g,h) = \{ i \in [\const{c:pre_bins} b^2] \mid f(a) \geq t(f) \wedge g(a) = i \wedge a \in A\}$ denote the indices hit in the domain $[\const{c:pre_bins} b^2]$ by the application of $g$ on the elements above the sub-sampling threshold. If $E_3(f,g,h)$, then $\size{\tilde{R}(f,g,h)} = \size{R(f)}$ and if $E_1(f,g,h) \wedge E_2(f,g,h)$ ,then $\size{R(f)} \leq b$ (see Eq.~\ref{eq:B_bound}). Recalling that $p(\psi)$ is the number of bins hit by the application of $k$-independent family from $\tilde{R}(\psi) \subseteq [\const{c:pre_bins} b^2]$ to $[b]$ we can apply Lemma~\ref{le:balls_and_bins_sum}. This implies:
\begin{align*}
    & \prob_{(f,g,h) \sim \Psi} \left( \textstyle \bigwedge_{i \in \{1,2,3\}} E_i(f,g,h) \wedge \size{p(f,g,h) - \rho(\size{R(f)})} \geq \textstyle \frac{\eps}{12} \size{R(f)} \right) \leq \\ 
    & \prob_{\psi \sim \Psi} \left( \size{\tilde{R}(\psi)} \leq b \wedge \size{p(\psi) - \rho\left(\size{\tilde{R}(\psi)}\right)} \geq \frac{\eps}{12} \size{\tilde{R}(\psi)} \right) \leq \\
    & \prob_{\psi \sim \Psi} \left( \size{\tilde{R}(\psi)} \leq b \wedge \size{p(\psi) - \rho\left(\size{\tilde{R}(\psi)}\right)} \geq 9 b^{-1/2} \size{\tilde{R}(\psi)} \right) \leq 2^{-6}
\end{align*}
where we used, that $b \geq 9^2 12^2 \eps^{-2}$ (i.e. $\const{c:delta} >= 9^2 12^2$.
\end{proof}
\begin{lemma} \label{le:median_single_1} Equation~\ref{eq:median_single} is true.
\end{lemma}
\begin{proof}
Let us start by observing that $E_1(\psi) \wedge E_2(\psi) \wedge E_4(\psi) \rightarrow \size{A^*(\psi) - \size{A}} \leq \eps \size{A}$. This is basically an error propagation argument. 
First note that by using Eq.~\ref{eq:B_bound}:
$p(f,g,h) \leq \rho(R(f)) + \frac{\eps}{12} \size{R(f)} \leq \rho(\frac{2}{3}b) + \frac{1}{12} \size{R(f)} \leq \frac{41}{60} b$.
Moreover, using the mean value theorem:
\[
    \size{ \rho^{-1}(p(f,g,h)) - \size{R(f)}} = (\rho^{-1})'(\xi) \size{ p(f,g,h) - \rho(\size{R(f)})} \leq \textstyle \frac{\eps}{3} \size{R(f)}
\]
for some $\xi$ between $\rho(\size{B(f)})$ and $p(f,g,h)$ where we can approximate $(\rho^{-1})'(\xi) < 4$. Hence:
\begin{eqnarray*}
    \size{ \rho^{-1}(p(f,g,h)) - 2^{-s(f)} \size{A} } & \leq & 
    \size{ \rho^{-1}(p(f,g,h)) - \size{R(f)} } + \size{ \size{R(f)} - 2^{-s(f)} \size{A}} \\ 
    & \leq & \frac{\eps}{3} \size{R(f)} + \size{ \size{R(f)}  - 2^{-s(f)} \size{A}} \\
    & \leq & \frac{\eps}{3} \size{R(f) - 2^{-s(f)} \size{A}} + \frac{\eps}{3} 2^{-s(f)} \size{A} + \size{ \size{R(f)}  - 2^{-s(f)} \size{A}} \\
    & \leq & \left( \frac{2\eps}{3} + \frac{\eps^2}{9} \right) 2^{-s(f)} \size{A} \leq \eps 2^{-s(f)} \size{A}
\end{eqnarray*}
It is also possible to deduce that $E_1(f,g,h) \rightarrow t(f) \geq \ceil{\ld (\size{A})} - \ld b \rightarrow s(f) \geq q_\mathrm{max}$.
Using Lemma~\ref{le:e_1} to \ref{le:e_4} we can conclude that Equation~\ref{eq:median_single_pre} is true. And the implications derived here show that then Equation~\ref{eq:median_single} must be true as well.
\end{proof}%
To extend the previous result to the case: $q \leq q_\mathrm{max}$, let us introduce the random variables:
\begin{align*}
    t_c(\psi, q) & := \max \{ \tau_1(\psi, q)[j] + q \mid j \in [b] \} - \ld b + 9 &
    s_c(\psi, q) & := \max(0, t_c(\psi, q)) \\
    p_c(\psi, q) & := \size{\{ j \in [b] \mid \tau_1(\psi, q)[j] + q \geq s_c(\psi,q) \}} &
    Y_c(\psi, q) & := 2^{s_c(\psi,q)} \rho^{-1}(p_c(\psi,q))
\end{align*}
These definitions $t_c$, $p_c$ and $Y_c$ correspond to the terms within the loop in the $\mathrm{estimate}$ function for arbitrary $q$.%
\begin{lemma} \label{le:median_single}
$\prob_{\psi \sim \Psi}\left( \exists q \leq q_\mathrm{max}. \size{Y_c(\psi,q) - \size{A}} > \eps \size{A} \right) \leq \frac{1}{16}$
\end{lemma}
\begin{proof}
It is possible to see that $t_c(\psi,q) = t(\psi)$ if $q \leq t(\psi)$. This is because $\tau_1(\psi,q) + q$ and $\tau_1(\psi,0)$ are equal except for values strictly smaller than $q$. With a case distinction on $t(\psi) \geq 0$ it is also possible to deduce that $s(\psi,q) = s(\psi)$ if $q \leq s(\psi)$.
Hence: $p_c(\psi, q) = p(\psi)$ and $Y_c(\psi,q) = Y(\psi)$ (for $q \leq s(\psi)$).
Thus this lemma is a consequence of Lemma~\ref{le:median_single_1}.
\end{proof}

The previous result established that each of the individual estimates is within the desired accuracy with a constant probability. The following establishes that the same is true for the median with a probability of $1- \frac{\delta}{2}$:

\begin{lemma}\label{le:median} $L := \prob_{\omega \in \Omega}\left( \exists q \leq q_\mathrm{max}. \size{\mathrm{estimate}(\tau_2(\omega,q)) - \size{A}} \geq \eps \size{A}\right) \leq \frac{\delta}{2}$
\end{lemma}
\begin{proof}
Because the median of a sequence will certainly be in an interval, if more than half of the elements are in it, we can approximate the left-hand side as:%
\begin{eqnarray*}
    L & \leq & \prob_{\omega \in \Omega}\left( \exists q \leq q_\mathrm{max}. \sum_{i \in [l]} [\size{Y(\omega_i,q) - \size{A}} \geq \eps \size{A}] \geq \frac{l}{2} \right) \\
    & \leq & \prob_{\omega \in \Omega}\left( \sum_{i \in [l]} [\exists q \leq q_\mathrm{max}. \size{Y(\omega_i,q) - \size{A}} \geq \eps \size{A}] \geq \frac{l}{2} \right) \\
    & \leq & \exp\left( - l \left(\frac{1}{2} \ln \left(\left(\frac{1}{16}+\frac{1}{16}\right)^{-1}\right) - 2 e^{-1}\right)\right) \leq \exp\left( - \frac{l}{4} \right) \leq \frac{\delta}{2}
\end{eqnarray*}%
The third inequality follows from Lemma~\ref{le:expander_chernoff} and \ref{le:median_single} as well as $\lambda \leq \frac{1}{16}$.
\end{proof}
We can now complete the proof of Theorem~\ref{th:overall}.
\begin{proof}[Proof of Theorem~\ref{th:overall}]
Follows from Lemma~\ref{le:cut_level} and the previous lemma, as well as the reasoning established in Equation~\ref{eq:overall_chain}.
\end{proof}

\subsection{Space Usage\label{sec:space_usage}}
It should be noted that the data structure requires an efficient storage mechanism for the levels in the bins. If we insist on reserving a constant number of bits per bins, the space requirement will be sub-optimal.
Instead we need to store the table values in a manner in which the number of bits required for a value $x$ is proportional to $\ln x$.
A simple strategy would be to store each value using a prefix-free universal code and concatenating the encoded variable-length bit strings.\footnote{Note that a vector of prefix-free values can be decoded even if they are just concatenated.} A well-known universal code for positive integers is the Elias-gamma code, which requires $2 \floor{ \ld x} + 1$ bits for $x \geq 1$~\cite{elias1975}. Since, in our case, the values are integers larger or equal to $(-1)$, they can be encoded using $2 \floor{ \ld (x+2) } + 1$ bits.\footnote{There are more sophisticated strategies for representing a sequence of variable-length strings that allow random access.~\cite{blandford2008}} (We are adding $2$ before encoding and subtracting after decoding.)
In combination with the condition established in the $\mathrm{compress}$ function of Algorithm~\ref{alg:main} the space usage for the table is thus $(2 \const{c:space_bound}+ 1) b l \in \bigo( b l ) \subseteq \bigo( \ln(\delta^{-1}) \eps^{2})$.
Additionally, the approximation threshold needs to be stored. This threshold is a non-negative integer between $0$ and $\ld n$ requiring $\bigo( \ln \ln n)$ bits to store.
In summary, the space required for the sketch is $\bigo( \ln(\delta^{-1}) \eps^{2} + \ln \ln n)$.
For the coin flips, we need to store a random choice from $\Omega$, i.e., we need to store $\ln (\size{\Omega})$ bits.
The latter is in %

\begin{eqnarray*}
    \bigo (\ln (\size{\Omega}) ) & \subseteq &  \bigo( \ln (\size{\Psi}) + l \ln (\lambda^{-1})) \\ 
    & \subseteq &  \bigo( \ln( \size{\mathcal G_2([n])}) + \ln( \size{ \mathcal H_2([n], [\const{c:pre_bins} b^2])}) + \ln( \size{ \mathcal H_k([\const{c:pre_bins} b^2],[b])}) + l^2 (\ln l)^3) \\ 
    & \subseteq & \bigo( \ln n + \ln n + k \ln (\eps^{-1}) + \ln (\delta^{-1})^{3} ) \\
    & \subseteq & \bigo( \ln n + \ln (\eps^{-1})^2 + \ln (\delta^{-1})^3)  \textrm{.}
\end{eqnarray*}
Overall the total space for the coin flips and the sketch is $\bigo( \ln(\delta^{-1}) \eps^{-2} + \ln n + \ln (\delta^{-1})^3)$.

\section{Extension to small failure probabilities\label{sec:ext_arb}}
The data structure described in the previous section has a space complexity that is close but exceeds the optimal $\bigo( \ln(\delta^{-1}) \eps^{-2} + \ln n)$. The main reason this happens is that, with increasing length of the random walk, the spectral gap of the expander is increasing as well --- motivated by the application of Lemma~\ref{le:deviation_bound} in Subsection~\ref{sec:cut_level}, with which we could establish that the cut-level could be shared between all tables. A natural idea is to restrict that.

If $\delta^{-1}$ is smaller than $\ln n$ the term $(\ln (\delta^{-1}))^3$ in the complexity of the algorithm is not a problem because it is dominated by the $\ln n$ term. If it is larger, we can split the table into sub-groups and introduce multiple cut-levels. Hence a single cut-level would be responsible for a smaller count of tables, and thus the requirements on the spectral gap would be lower. (See also Figure~\ref{fig:outer_inner_state}).

A succinct way to precisely prove the correctness of the proposal is to repeat the previous algorithm, which has only a single shared cut-level, in a black-box manner for the same universe size and accuracy but for a higher failure probability. The seeds of each repetition are selected again using an expander walk.  Here the advantage of Lemma~\ref{le:expander_chernoff} is welcome, as the inner algorithm needs to have a failure probability depending on $n$ --- the natural choice is $(\ln n)^{-1}$. This means the length of the walk of the inner algorithm matches the number of bits of the cut-level $\bigo(\ln \ln n)$. The repetition count of the outer algorithm is then $\bigo\left(\frac{\ln (\delta^{-1})}{\ln \ln n}\right)$. Note that the total repetition count is again $\bigo( \ln (\delta^{-1}))$.
\begin{theorem}\label{th:final_result}
Let $n > 0$, $0 < \eps < 1$ and $0 < \delta < 1$. Then there exists a cardinality estimation data structure for the universe $[n]$ with relative accuracy $\eps$ and failure probability $\delta$ with space usage $\bigo(\ln(\delta^{-1}) \eps^{-2} + \ln n)$.
\end{theorem}%
\begin{proof}
If $\delta^{-1} < \ln n$, then the result follows from Theorem~\ref{th:overall} and the calculation in Subsection~\ref{sec:space_usage}. Moreover, if $n < \exp(e^5)$, then the theorem is trivially true, because there is an exact algorithm with space usage $\exp(e^5) \in \bigo(1)$. Hence we can assume $e^5 \leq \ln n \leq \delta^{-1}$. Let $\Omega^*$, $\mathrm{single}^*$, $\mathrm{merge}^*$ and $\mathrm{estimate}^*$ denote the seed space and the API of Algorithm~\ref{alg:main} for the universe $[n]$, relative accuracy $\eps$ and failure probability $\delta^* := (\ln n)^{-1}$. 
Moreover, let $m := \left\lceil 4 \frac{\ln (\delta^{-1})}{\ln \ln n} \right\rceil$ --- the plan is to show that with these definitions Algorithm~\ref{alg:main_ext} fulfills the conditions of this theorem.
Let $\nu(\theta,A)[i] := \tau^*(\theta_i, A)$ for $i \in [m]$ and $\theta \in \Theta := U(\mathcal E(\Omega^*,\delta^*,m))$. Then it is straightforward to check that:
\begin{align*}
    \mathrm{single}(\theta,x) & = \nu(\theta,\{x\}) &
    \mathrm{merge}(\nu(\theta,A),\nu(\theta,B)) & = \nu(\theta,A \cup B)
\end{align*}
for $x \in [n]$ and $\emptyset \neq A, B \subseteq [n]$ taking into account Lemma~\ref{le:histind}.
Hence the correctness follows if:
$\prob_{\theta \in \Theta}( \size{\mathrm{estimate}(\nu(\theta,A)) - \size{A}} > \eps \size{A} ) \leq \delta$.
Because the estimate is the median of the individual estimates, this is true if at least half of the individual estimates are in the desired range.
Similar to the proof of Lemma~\ref{le:median} we can apply Lemma~\ref{le:expander_chernoff}. 
This works if
\[
    \exp\left( -m \left(\frac{1}{2} \ln \left((\delta^* + \delta^*)^{-1}\right) - 2e^{-1}\right)\right) \leq \delta
\]
which follows from  $m \geq 4 \ln (\delta^{-1}) (\ln \ln n)^{-1}$ and $\ln \ln n \geq 5$.
The space usage for the seed is: $\ln \size{\Theta} \in \bigo(\ln n + \ln (\eps^{-1})^2 + (\ln ((\delta^*)^{-1}))^3 + m \ln ((\delta^*)^{-1})) \subseteq
\bigo (\ln n + \ln (\eps^{-1})^2 + \ln (\delta^{-1}))$.
And the space usage for the sketch is: $\bigo( m \ln ((\delta^*)^{-1}) \eps^{-2} + m \ln \ln n ) \subseteq \bigo( \ln (\delta^{-1}) \eps^{-2} + \ln \ln n)$.
\end{proof}
\begin{algorithm}[H]%
\begin{pseudo*}[indent-mark]
\kw{function} $\mathrm{init}()$ : $\Theta$ \\+
    \kw{return} $\mathrm{random} \, U(\Theta)$ \\-
\\
\kw{function} $\mathrm{single}(x : U, \theta : \Theta)$ : $\mathcal S$ \\+
    $D[i] = \mathrm{single}^*(x, \theta_i)$ \kw{for} $i \in [m]$\\
    \kw{return} $D$ \\-
\\
\kw{function} $\mathrm{merge}(D_a : \mathcal S, D_b : \mathcal S)$ : $\mathcal S$ \\+
    $D[i] \gets \mathrm{merge}^*(D_a[i],D_b[i])$ \kw{for} $i \in [m]$ \\
    \kw{return} $D$ \\-
\\
\kw{function} $\mathrm{estimate}(D : \mathcal S)$ : $\mathbb R$ \\+
    $Y_i \gets \mathrm{estimate}^*(D[i])$ \kw{for} $i \in [m]$ \\
    \kw{return} $\mathrm{median} (Y_0,\ldots, Y_{m-1})$ %
\end{pseudo*}%
\caption{Algorithm for $0 < \delta < (\ln n)^{-1}$} %
\label{alg:main_ext}%
\end{algorithm}
\begin{figure}[h!]
\centering

\includegraphics{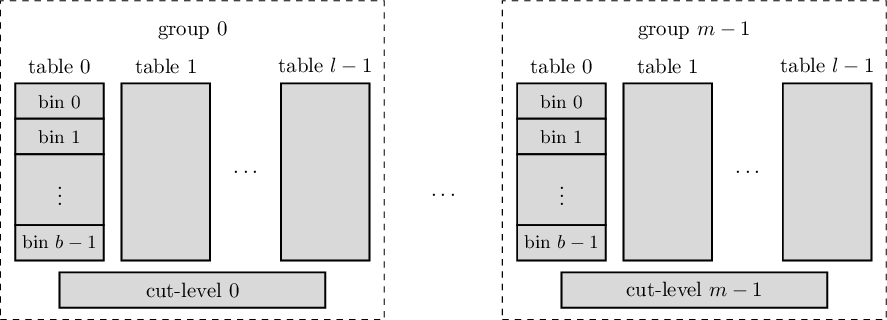}

\caption{Schematic representation of the states of Algorithm~\ref{alg:main_ext} with $m \in \bigo\left(\frac{\ln (\delta^{-1})}{\ln \ln n}\right)$ repetitions of the inner algorithm. The inner algorithm uses $b \in \bigo(\eps^{-2})$ bins and $l \in \bigo(\ln \ln n)$ tables.}
\label{fig:outer_inner_state}
\end{figure}

\section{Optimality\label{sec:counter_example}}
The optimality of the algorithm introduced by B\l{}asiok~\cite{blasiok2020} follows from the lower bound established by Jayram and Woodruff~\cite[Theorem 4.4]{jayram2013}. The result (as well as its predecessors~\cite{alon1999, woodruff2004}) follows from a reduction to a communication problem. This also means that their theorem is a lower bound on the information the algorithm needs to retain between processing successive stream elements.

It should be noted that, if additional information is available about the distribution of the input, the problem becomes much easier.
Indeed with such assumptions it is even possible to introduce algorithms that can approximate the cardinality based on observing only a fraction of the input, so the upper bound established in the previous section and the lower bounds discussed here are with respect to algorithms, that work for all inputs.\footnote{The probabilistic nature of the correctness condition is only with respect to the internal random bits used.}

An immediate follow-up question to Theorem~\ref{th:final_result} is whether the space usage is also optimal in the distributed setting. Unfortunately, this question is not as well posed as it sounds. One interpretation would be to ask whether there is a randomized data structure that fulfills the API described at the beginning of Section~\ref{sec:alg}, i.e., with the four operations: init, single, merge and estimate, fulfilling the same correctness conditions, requiring $o(\ln (\delta^{-1} )\eps^{-2} + \ln n)$. For that question the answer is no, because such a data structure can be converted into a sequential streaming algorithm: Every time a new stream element is processed, the new state would be computed by obtaining the sketch of the new element using the single operation and merging it with the pre-existing state using the merge operation. (See also the first mode of operation presented in Figure~\ref{fig:dist_stream_example}.)

A more interesting question is, if there is a less general algorithm that works in the distributed streams model. Let us assume there are $p$ processes, each retaining $m$ stream elements, and they are allowed to communicate at the beginning, before observing the stream elements, and after observing all stream elements. Here, let us assume that the processes know how many processes there are and also how many stream elements each process owns. Even with these relaxed constraints, the number of bits that each process will need to maintain will be the same as the minimum number of bits of a sequential streaming solution.   
This follows by considering a specific subset of the input set where except for process $0$, the stream elements on all the other processes are equal to the last stream element of process $0$. In particular, the information the processes $1,2, \dots, p-1$ have is $0$ bits from the perspective of process $0$. If our distributed hypothetical algorithm is correct, it can only be so if the worst-case space usage per process is $\Omega(\ln (\delta^{-1} )\eps^{-2} + \ln n)$.

It should be noted that more relaxed constraints, for example, if the processes are allowed to communicate multiple times, after having observed some of the stream elements, prevent the previous reduction argument. And there will be more efficient solutions. Similar things happen, if assumptions about the distribution of the input are made.

\section{Runtime\label{sec:runtime}}
The function compress in Algorithm~\ref{alg:main}, which is being used as an internal operation within the single and merge operations is described in a way that allows verifying its correctness properties easily, but as an algorithm it has sub-optimal runtime. In the following, I want to introduce an alternative faster implementation with the same behavior.

Let us recall that the function repeatedly decrements every (non-negative) table entry and increments the cut-off level until the condition
\begin{equation}\label{eq:compress}
    L := \sum_{i \in [l], j \in [b]} \floor{\ld (B[i,j]+2)} > \const{c:space_bound} b l
\end{equation}
is fulfilled. If the number of iterations in the loop --- the minimum value that the table entries need to be decreased by --- is known, the while loop can be removed. This results in an algorithm of the following form:
\begin{pseudo*}[indent-mark]
\kw{function} $\mathrm{compress}((B,q) : \mathcal S)$ : $\mathcal S$ \\+
  $\Delta \gets \mathrm{find{\isacharunderscore}required{\isacharunderscore}cutoff}(B)$ \\
  $B[i,j] \gets \max(B[i,j]-\Delta,-1)$ \kw{for} $i \in [l], j \in [b]$ \\
  $q \gets q + \Delta$ \\
  \kw{return} $(B,q)$ \\-
\end{pseudo*}
The function $\mathrm{find{\isacharunderscore}required{\isacharunderscore}cutoff}$ is a dynamic programming algorithm. It starts by computing the minimum amount by which the left-hand side of Eq.~\ref{eq:compress} has to be reduced. Then it computes a temporary table, with which it is possible to determine the effect of every possible $\Delta$ on $L$. To understand how that works, let us first note that the contribution of a single table entry $B[i,j]$ to $L$ will change only if $\floor{\ld (B[i,j]+2)}$ is affected, which splits the possible $\Delta$ values into distinct consecutive intervals. For example: If $B[i,j] = 23$, then any $\Delta$ below $9$ will not affect $\floor{\ld (B[i,j]+2)}$. If $\Delta$ is between $10$ and $17$, then the contribution of $\floor{\ld (B[i,j]+2)}$ will decrease by one. If $\Delta$ is between $18$ to $21$, it will decrease by two, etc. All of that can be kept track off more efficiently using a sequence $\chi$ which describes the relative effect of a $\Delta$ compared to $\Delta-1$, i.e., the discrete derivative of the function we are looking for. For our example this means that $\chi$ will be $1$ for the values: $10, 18, 22, 24$ and $0$ otherwise. It is, of course straightforward to cumulatively determine $\chi$ for the entire table.
\begin{pseudo*}[indent-mark]
\kw{function} $\mathrm{find{\isacharunderscore}required{\isacharunderscore}cutoff}(B)$ : $\mathbb N$ \\+
  $R \gets \sum_{i \in [l], j \in [b]} \floor{\ld (B[i,j]+2)} - \const{c:space_bound} b l$ \\
  $\chi[i] \gets 0$ \kw{for} $i \in [\ceil{\ld n}]$ \\
  \kw{for} $(i,j) \in [l] \times [b]$ \\+
    $x \gets B[i,j]+2$ \\
    $\mathrm{inc}(\chi[x - (2^k-1)])$ \kw{for each} $1 \leq k \leq \floor{\ld x}, k \in \mathbb N$\\-
  $\Delta \gets 0$ \\
  \kw{while} $R > 0$: \\+
    $\mathrm{inc}(\Delta)$ \\
    $R \gets R - \chi[\Delta]$ \\-
  \kw{return} $\Delta$ \\
\end{pseudo*}
In the last step, the algorithm determines the smallest $\Delta$ fulfilling Eq.~\ref{eq:compress} using the function $\chi$, i.e., the length of the smallest prefix of $\chi$ whose sum surpasses $R$. To estimate the runtime of the above compression algorithm and the resulting merge and estimate operations, it makes sense to first obtain a bound on the left-hand side of Eq~\ref{eq:compress}, for any possible \emph{input} of the compress operation.
\begin{itemize}
    \item For the single operation: $L \in \bigo( \ln(\delta^{-1}) \ln \ln n ) \subseteq \bigo(\ln n)$. 
    \item For the merge operation: $L \in \bigo( \ln(\delta^{-1}) \eps^{-2})$.
\end{itemize}
The first observation follows from the definition of $\mathrm{single}_1$ in Algorithm~\ref{alg:main}. Note that this is within the context of the inner algorithm (Section~\ref{sec:alg}), where it is correct to assume $\delta^{-1} \leq \ln n$. For the merge operation, this follows from the fact that the initial $\mathrm{merge}_1$ can at most double the space usage of its inputs, where for each input Eq.~\ref{eq:compress} can be assumed.
On the other hand it is easy to check that the runtime of the new compress function is in $\bigo(L + \ln(\delta^{-1}) \eps^{-2} + \ln n)$ in the word RAM model for a word size $w \in O(\max (\ln n, \ln (\eps^{-1}), \ln \ln(\delta^{-1})))$. In summary, the operations merge and single require $O(\ln(\delta^{-1}) \eps^{-2} + \ln n)$ operations.

A practical implementation of the estimate function introduced in Algorithm~\ref{alg:main} requires an approximation of $\rho^{-1}(x)$. This can be done by increasing the parameter $b$ by a factor of $4$ (and the parameter $k$ accordingly, since it is defined in terms of $b$) and computing an approximation of 
$\rho^{-1}(x)$ with an error of $\eps/2$ (in the range $0 \leq x \leq \frac{41}{60}b$)\footnote{Because of Lemma~\ref{le:median_single_1}, it is enough to approximate $\rho^{-1}$ only within this range.}. In combination the resulting algorithm has again a total relative error of $\eps$. For such an implementation the number of operations is asymptotically $O(\ln(\delta^{-1}) \eps^{-2} + \ln n)$. 

It is straightforward to extend the same result to the extended solution derived in Section~\ref{sec:ext_arb}.

\section{Conclusion\label{sec:conclusion}}
A summary of this work would be that for the space complexity of cardinality estimation algorithms, there is no gap between the distributed and sequential streaming models. Moreover, it is possible to solve the problem optimally (in either model) with expander graphs and hash families without using code-based extractors (as they were used in previous work). The main algorithmic idea is to avoid using a separate rough estimation data structure for quantization (cut-off); instead, the cut-off is guided by the space usage. During the estimation step at the end, an independent rough estimate is still derived, but it may be distinct from the cut-off reached at that point. This is the main difference between this solution and the approach by Kane et al.~\cite{kane2010}. The main mathematical idea is to take the tail estimate based on the Kullback-Leibler divergence for random walks on expander graphs, first noted by Impagliazzo and Kabanets~\cite[Th.~10]{impagliazzo2010} seriously. With which, it is possible to achieve a failure probability of $\delta$ using $\bigo\left(\frac{\ln (\delta^{-1})}{\ln ((\delta^*)^{-1})}\right)$ repetitions of an inner algorithm with a failure probability $\delta^* > \delta$. Note that the same cannot be done with the standard Gillman-type Chernoff~\cite{gillman1998} bounds. This allows the two-stage expander construction that we needed. As far as I can tell, this strategy is new and has not been used before.

B\l{}asiok~\cite{blasiok2020} and Kane et al.~\cite{kane2010} also discuss strong tracking properties for the sequential streaming algorithm. Their methods do not scale into the distributed stream model, because the possible number of reached states is exponentially larger than the number of possible states in the sequential case. An interesting question is whether there are different approaches for the distributed streams model or complexity bounds with respect to the number of participating processes or total number of stream elements, with which strong-tracking properties can be derived.

Another interesting question is whether the two-stage expander construction can somehow be collapsed into a single stage. For that, it is best to consider the following non-symmetric aggregate:
\[
    \prob_{\omega \in \mathcal E (\mathcal E (S, \exp(-l (\ln l)^3), l), \exp(-l/m), m)} \left(\sum_{i \in [m]} \left[\sum_{j \in [l]} X(\omega_{ij}) \geq \const{c:dev_bound} \right] \geq \frac{m}{2} \right) \leq \exp(-\bigo(lm))
\]
where $X$ may be an unbounded random variable with, e.g., sub-gaussian distribution. Indeed, the bound on the count of too-large cut-off values from Algorithm~\ref{alg:main_ext} turns out to be a tail estimate of the above form. I tried to obtain such a bound using only a single-stage expander walk but did not succeed without requiring too large spectral gaps, i.e., with $\lambda^{-1} \in \bigo(1)$ for $m \ll l$. There is a long list of results on more advanced Chernoff bounds for expander walks~\cite{agrawal2019,lezaud1998,naor2019,rao2018,rao2017,wagner2006} and investigations into more general aggregation (instead of summation) functions~\cite{cohen2021, garg2017, golowich2022, golowich2022_2, paulin2015, reingold2013}, but I could not use any of these results/approaches to avoid the two-stage construction.
This suggests that either there are more advanced results to be found or multi-stage expander walks are inherently more powerful than single-stage walks.
\bibliography{main}
\appendix
\newpage
\section{Proof of Lemma~\ref{le:deviation_bound}\label{apx:proof_dev_bound}}
\deviationboundstatement*
\begin{proof}
Let $\mu_k := \expect_{v \sim V} [e^k \leq f(v)] \leq \exp(-e^k k^3)$ for $k \geq 3$. We will show
\begin{equation}
\label{eq:deviation_bound}
    L_k := \prob_{w \sim \mathrm{Walk}(G,l)}\left( \sum_{i \in [l]} [e^k \leq f(w_i)] \geq l e^{-k} k^{-2} \right) \leq \exp(-l-k+2) \textrm{ for all } k \geq 3
\end{equation}
by case distinction on the range of $k$:

\emph{Case $k \geq \max(\ln l,3)$}:
In this case the result follows using Markov's inequality. Note that the random walk starts from and remains in the stationary distribution, and thus for any index $i \in [l]$ the distribution of the $i$-th walks step $w_i$ will be uniformly distributed over $V$, hence: 
\begin{eqnarray*}
    L_k & \leq & e^k k^2 l^{-1} \expect_{w \sim \mathrm{Walk}(G,l)} \textstyle \sum_{i \in [l]} [e^k \leq f(w_i)] = e^k k^2 \expect_{v \sim V} [e^k \leq f(v)] \\
    & \leq & e^k k^2 \exp(-e^k k^3 ) = \exp( k + 2 \ln k - e^k k^3 ) \leq \exp( 2 k - e^k (k^2 + 2) ) \\ 
    & \leq & \exp( 2 k - e^k k^2 - e^k - e^k ) \leq \exp ( -l -k + 2 )
\end{eqnarray*}
Here we use that $k^3 \geq k^2 + 2$ and $e^k \geq k$ for $k \geq 3$ and $e^k \geq l$.

\emph{Case $3 \leq k < \ln l$}:
Then we have
\begin{eqnarray*}
    L_k & \leq & \exp \left( -l (e^{-k} k^{-2} \ln ((\mu_k + \lambda)^{-1}) - 2e^{-1}) \right) \textrm{ using Lemma~\ref{le:expander_chernoff}} \\
        & \leq & \exp \left( -l (e^{-k} k^{-2} (e^k k^3 - \ln 2) - 2e^{-1}) \right) \leq \exp \left( -l ( k - e^{-k} k^{-2} \ln 2 - 2e^{-1}) \right) \\
        & \leq & \exp \left( -l ( k - 1) \right) \leq \exp \left( -l - k + 2 \right)
\end{eqnarray*}
Concluding the proof of Eq.~\ref{eq:deviation_bound}.

Note that:
\begin{eqnarray*}
    \sum_{i \in [l]} f(w_i) & \leq & e^2 l + \sum_{i \in [l]} \sum_{k \geq 2} e^{k+1} [e^k \leq f(w_i) < e^{k+1}] \\
    & \leq & e^2 l + \sum_{i \in [l]} \left( \sum_{k \geq 2} e^{k+1} [e^k \leq f(w_i)] - \sum_{k \geq 2} e^{k+1} [e^{k+1} \leq f(w_i)] \right) \\
    & \leq & (e^2 + e^3) l + (e-1) \sum_{i \in [l]} \left( \sum_{k \geq 3} e^k [e^k \leq f(w_i)]  \right)
%    & \leq & (e^2 + e^3) l + (e-1) \sum_{i \in [l], k \geq 2} e^k [e^k \leq f(w_i)]
\end{eqnarray*}
Hence:
\begin{eqnarray*}
    \prob_{w \sim \mathrm{Walk}(G,l)} \left( \sum_{i \in [l]} f(w_i) \geq \const{c:dev_bound} l \right) & \leq &
    \prob_{w \sim \mathrm{Walk}(G,l)} \left( \sum_{k \geq 3, i \in [l]} e^k [e^k \leq f(w_i)] \geq l \right) \\
    & \leq & \prob_{w \sim \mathrm{Walk}(G,l)} \left( \bigvee_{k \geq 3} \sum_{i \in [l]} [e^k \leq f(w_i)] \geq l e^{-k} k^{-2} \right) \\
    & \leq & \sum_{k \geq 3} L_k \leq \sum_{k \geq 3} \exp \left( -l - k + 2 \right) \leq \exp(-l) \textrm{.\qedhere}
\end{eqnarray*}
\end{proof}

\section{Balls and Bins}
Let $\Omega = U( [r] \rightarrow [b])$ be the uniform probability space over the functions from $[r]$ to $[b]$ for $b \geq 1$ and $0 \leq r \leq b$ and let 
$X(\omega) = \size{\omega([r])}$ be the size of the image of such a function. This models throwing $r$ balls into $b$ bins independently, where $X$ is the random variable counting the number of hit bins. Moreover, let $E_i(\omega) = \{ \omega \mid i \in \omega([r]) \}$ be the event that the bin $i$ was hit.
Note that $X(\omega) = \sum_{i \in [b]} E_i(\omega)$.
And we want to show that
\begin{align*}
    \expect_{\omega \sim \Omega} X(\omega) & = b \left(1-\left(1-\frac{1}{b}\right)^r\right) & 
    \var_{\omega \sim \Omega} X(\omega) & \leq \frac{r(r-1)}{b} 
\end{align*}

\begin{lemma}\label{le:ind_bin_balls_exp}
    $\expect_{\omega \sim \Omega} X(\omega) = b \left(1-\left(1-\frac{1}{b}\right)^r\right)$
\end{lemma}%
\begin{proof}
First note that:
\[
    \prob (\neg E_i) = \prob_{\omega \sim \Omega} \left\{ \omega \mid \omega([r]) \subseteq [b] \setminus \{ i \} \right\} = \left(1-\frac{1}{b}\right)^r  
\]
which can be seen by counting the number of functions from $[r]$ to $[b] \setminus \{ i\}$.
Hence:
\[
    \expect X = \sum_{i \in [b]} 1 - \prob (\neg E_i) = b \left(1 - \left( 1-\frac{1}{b}\right)^r \right) \textrm{\qedhere}
\]
\end{proof}
\Needspace{20\baselineskip}

\begin{lemma}\label{le:ind_bin_balls_var}
    $\var_{\omega \sim \Omega} X(\omega) \leq \frac{r(r-1)}{b}$
\end{lemma}

\begin{proof}
Note that for $r \leq 1$: $\var X = 0$ because $X$ is constant. For $r \geq 2$:
\begin{eqnarray*}
    \var X & = & \expect X^2 - (\expect X)^2 = \sum_{i,j \in [b]} \prob (E_i \wedge E_j) - (\prob E_i) (\prob E_j) \\
    & = & \sum_{i,j \in [b]} (1- \prob(\neg E_i) - \prob(\neg E_j) + \prob(\neg E_i \wedge \neg E_j)) - (1-\prob(\neg E_i)) (1-\prob(\neg E_j)) \\
    & = & \sum_{i,j \in [b]} \prob(\neg E_i \wedge \neg E_j) - \prob(\neg E_i) \prob(\neg E_j) \\
    & = & \sum_{i \neq j \in [b]} \left(1-2 b^{-1}\right)^r - \left(1-b^{-1}\right)^{2r} + \sum_{i \in [b]} \left( 1-b^{-1} \right)^r -  \left( 1-b^{-1} \right)^{2r} \\
    & = & b (b-1) \left[ \left(1-2 b^{-1}\right)^r - \left(1-b^{-1}\right)^{2r} \right] +  b \left[ \left( 1-b^{-1} \right)^r -  \left( 1-b^{-1} \right)^{2r} \right] \\
    & = & b^2 \left[ \left(1-2 b^{-1}\right)^r - \left(1-b^{-1}\right)^{2r} \right] +  b \left[ \left( 1- b^{-1} \right)^r -  \left( 1-2 b^{-1} \right)^{r} \right] \\
    & = & -r \xi_1^{r-1} + r \xi_2^{r-1} \textrm{ for some } \xi_1 \in \left(1-2b^{-1},\left(1-b^{-1}\right)^2\right), \xi_2 \in \left(1-2b^{-1},1-b^{-1}\right) \\
    & \leq & -r \left(1-2b^{-1}\right)^{r-1} + r \left(1-b^{-1}\right)^{r-1} \\
    & = & r (r-1) b^{-1} \xi_3^{r-2} \textrm{ for some } \xi_3 \in \left(1-2b^{-1},1-b^{-1}\right) \\
    & \leq & \frac{r(r-1)}{b} \textrm{\qedhere}
\end{eqnarray*}%
\end{proof}
The lines where the variables $\xi_i$ were introduced follow from the application of the mean value theorem.
The above is a stronger version of the result by Kane et al.~\cite{kane2010}[Lem.~1]. Their result has the restriction that $r \geq 100$ and a superfluous factor of $4$.

Interestingly, it is possible to obtain a similar result for $k$-independent balls into bins. For that let $\Omega'$ be a probability space of functions from $[r]$ to $[b]$ where 
\[
    \prob_{\omega \sim \Omega'} \left( \bigwedge_{i \in I} \omega(i) = x(i) \right) = r^{-\size{I}}
\]
for all $I \subset [r]$, $\size{I} \leq k$ and all $x : I \rightarrow [b]$.
As before let us denote $X'(\omega) := \size{\omega([r])}$ the number of bins hit by the $r$ balls. Then the expectation (resp. variance) of $X'$ approximates that of $X$ with increasing independence $k$, more precisely: 
\begin{lemma}
\label{le:approx_bin_balls}
If $\varepsilon \leq e^{-2}$ and $k \geq 1 + 5 \ln (b \eps^{-1}) (\ln (\ln (b \eps^{-1})))^{-1}$ then: 
\begin{align*}
    \size{\expect_{\omega' \in \Omega'} X'(\omega') - \expect_{\omega \in \Omega} X(\omega)} & \leq \varepsilon r &
    \size{\var_{\omega' \in \Omega'} X'(\omega') - \var_{\omega \in \Omega} X(\omega)} & \leq \varepsilon^2 \, \textrm{.}
\end{align*}
\end{lemma}
This has been shown\footnote{Without the explicit constants mentioned in here.} by Kane et al.~\cite{kane2010}[Lem.~2]. The proof relies on the fact that
$X = \sum_{i \in [b]} \max(1,Y_i)$ where $Y_i$ denotes the random variable that counts the number of balls in bin $i$.
It is possible to show that $\expect (Y_i)^j = \expect (Y'_i)^j$ for all $j \leq k$ (where $Y'_i$ denotes the same notion over $\Omega'$).
Their approach is to approximate $\max(1,\cdot)$ with a polynomial $g$ of degree $k$.
Since $\expect g(Y_i) = \expect g(Y'_i)$ they can estimate the distance between $\expect X$ and $\expect X'$ by bounding the expectation of each approximation error: $g(Y_i) - \max(1,Y_i)$. Obviously, larger degree polynomials (and hence increased independence) allow better approximations. The reasoning for the variance is analogous.

\begin{lemma}\label{le:balls_and_bins_sum}
If $k \geq \const{c:approx_bin_balls_1} \ln b + \const{c:approx_bin_balls_2}$ then:
\[
    L := \prob_{\omega' \in \Omega'} \left( \size{ X'(\omega') - \rho(r)} > 9 b^{-1/2} r \right) \leq 2^{-6}
\]
\end{lemma}
\begin{proof}
This follows from Lemma~\ref{le:ind_bin_balls_exp}, \ref{le:ind_bin_balls_var} and the previous lemma for $\eps = \min (e^{-2}, b^{-1/2})$ in particular:
$\var X' \leq \var X + \frac{1}{b} \leq \frac{r^2}{b}$ and hence:
\begin{eqnarray*}
    L & \leq & \prob_{\omega' \in \Omega'} \left( \size{ X'(\omega') - \expect X'} + \size{\expect X' - \rho(r)} \geq 9 b^{-1/2} r \right) \\
    & \leq & \prob_{\omega' \in \Omega'} \left( \size{ X'(\omega') - \expect X'} + b^{-1/2} r \geq 9 b^{-1/2} r \right) \\
    & \leq & \prob_{\omega' \in \Omega'} \left( \size{ X'(\omega') - \expect X'} \geq 8 b^{-1/2} r \right) \\
    & \leq & \prob_{\omega' \in \Omega'} \left( \size{ X'(\omega') - \expect X'} \geq 8 \sqrt{ \var X'} \right) \leq 2^{-6}
\end{eqnarray*}
where the last line follows from Chebychev's inequality.
\end{proof}
\section{Table of Constants}
\begin{table}[H]
\centering
\begin{tabular}{l l | l l}
\toprule
Constant & References & Constant & References \\ 
\midrule
$\makeconst{c:dev_bound} := e^2 + e^3 + (e-1)$ & Lemma~\ref{le:deviation_bound} &
$\makeconst{c:approx_bin_balls_1} := \frac{15}{2}$ & Lemma~\ref{le:approx_bin_balls} \\
$\makeconst{c:approx_bin_balls_2} := 16$ & Lemma~\ref{le:approx_bin_balls} &
$\makeconst{c:delta} := 3^2 2^{23}$ & Lemma~\ref{le:e_1} and \ref{le:e_4} \\
$\makeconst{c:space_bound} := \ceil{\const{c:dev_bound} + 3} = 33$ & Lemma~\ref{le:cut_level} &
$\makeconst{c:eps} := 4$ & Lemma~\ref{le:median} \\
$\makeconst{c:pre_bins} := 2^5$ & Lemma~\ref{le:e_3} \\%&
\bottomrule
\end{tabular}
\caption{Table of Constants}\label{tab:consts}
\end{table}
\section{Formalization\label{sec:formalization}}
As mentioned in the introduction the proofs in this work have been machine-checked using Isabelle. They are available~\cite{Distributed_Distinct_Elements-AFP, Expander_Graphs-AFP} in the AFP (Archive of Formal Proofs)~\cite{afp} --- a site hosting formal proofs verified by Isabelle. Table~\ref{tab:formalization} references the corresponding facts in the AFP entries. The first column refers to the lemma in this work. The second is the corresponding name of the fact in the formalization. The formalization can be accessed in two distinct forms: As a source repository with distinct theory files, as well as two ``literate-programming-style'' PDF documents with descriptive text alongside the Isabelle facts (optionally with the proofs). The latter is much more informative. The third column of the table refers to the file name
\footnote{\texttt{Distributed\_Distinct\_Elements} is abbreviated by \texttt{DDE} and \texttt{Without} with \texttt{WO}.}
of the corresponding source file, while the last column contains the reference of the AFP entry, including the section in the PDF versions.
\begin{longtable}{l p{4.5cm} l l l}
\toprule
Lemma & Formalized Entity & Theory & Src. \\
\midrule
Thm.~\ref{th:expander_chernoff} & \multicolumn{3}{p{11.2cm}}{This theorem from Impagliazzo and Kabanets was stated for motivational reasons and is never used in any of the following results, hence it is not formalized.}\\
Thm.~\ref{th:expander_walk_hitting} & \textbf{theorem} \emph{hitting{\isacharunderscore}property} & \isatheory{Expander_Graphs_Walks} & \cite[\S 9]{Expander_Graphs-AFP} \\
Thm.~\ref{th:expander_chernoff_improved} & \textbf{theorem} \emph{kl{\isacharunderscore}chernoff{\isacharunderscore}property} & \isatheory{Expander_Graphs_Walks}& \cite[\S 9]{Expander_Graphs-AFP} \\
Lem.~\ref{le:expander_chernoff} & \textbf{lemma} \emph{walk{\isacharunderscore}tail{\isacharunderscore}bound} & \isatheory{DDE_Tail_Bounds} & \cite[\S 5]{Distributed_Distinct_Elements-AFP} \\
Lem.~\ref{le:deviation_bound} & \textbf{lemma} \emph{deviation{\isacharunderscore}bound} & \isatheory{DDE_Tail_Bounds} & \cite[\S 5]{Distributed_Distinct_Elements-AFP} \\
Lem.~\ref{le:histind} (1) & \textbf{lemma} \emph{single{\isacharunderscore}result} & \isatheory{DDE_Inner_Algorithm} & \cite[\S 6]{Distributed_Distinct_Elements-AFP} \\
Lem.~\ref{le:histind} (2) & \textbf{lemma} \emph{merge{\isacharunderscore}result} & \isatheory{DDE_Inner_Algorithm} & \cite[\S 6]{Distributed_Distinct_Elements-AFP} \\
Lem.~\ref{le:cut_level} & \textbf{lemma} \emph{cutoff{\isacharunderscore}level} & \isatheory{DDE_Cutoff_Level} & \cite[\S 8]{Distributed_Distinct_Elements-AFP} \\
Lem.~\ref{le:e_1} & \textbf{lemma} \emph{e{\isacharunderscore}1} & \isatheory{DDE_Accuracy_WO_Cutoff} & \cite[\S 7]{Distributed_Distinct_Elements-AFP} \\
Lem.~\ref{le:e_2} & \textbf{lemma} \emph{e{\isacharunderscore}2} & \isatheory{DDE_Accuracy_WO_Cutoff} & \cite[\S 7]{Distributed_Distinct_Elements-AFP} \\
Lem.~\ref{le:e_3} & \textbf{lemma} \emph{e{\isacharunderscore}3} & \isatheory{DDE_Accuracy_WO_Cutoff} & \cite[\S 7]{Distributed_Distinct_Elements-AFP} \\
Lem.~\ref{le:e_4} & \textbf{lemma} \emph{e{\isacharunderscore}4} & \isatheory{DDE_Accuracy_WO_Cutoff} & \cite[\S 7]{Distributed_Distinct_Elements-AFP} \\
Lem.~\ref{le:median_single_1} & \textbf{lemma} \newline \hphantom{i} \emph{accuracy{\isacharunderscore}without{\isacharunderscore}cutoff} & \isatheory{DDE_Accuracy_WO_Cutoff} & \cite[\S 7]{Distributed_Distinct_Elements-AFP} \\
Lem.~\ref{le:median_single} & \textbf{lemma} \emph{accuracy{\isacharunderscore}single} & \isatheory{DDE_Accuracy} & \cite[\S 9]{Distributed_Distinct_Elements-AFP} \\
Lem.~\ref{le:median} & \textbf{lemma} \emph{estimate{\isacharunderscore}result{\isacharunderscore}1} & \isatheory{DDE_Accuracy} & \cite[\S 9]{Distributed_Distinct_Elements-AFP} \\
Thm.~\ref{th:overall} & \textbf{lemma} \emph{estimate{\isacharunderscore}result} & \isatheory{DDE_Accuracy} & \cite[\S 9]{Distributed_Distinct_Elements-AFP} \\ 
Thm.~\ref{th:final_result} (1) & \textbf{theorem} \emph{correctness} & \isatheory{DDE_Outer_Algorithm} & \cite[\S 10]{Distributed_Distinct_Elements-AFP} \\
Thm.~\ref{th:final_result} (2) & \textbf{theorem} \emph{space{\isacharunderscore}usage} & \isatheory{DDE_Outer_Algorithm} & \cite[\S 10]{Distributed_Distinct_Elements-AFP} \\
Thm.~\ref{th:final_result} (3) & \textbf{theorem} \newline \hphantom{i} \emph{asymptotic{\isacharunderscore}space{\isacharunderscore}complexity} & \isatheory{DDE_Outer_Algorithm} & \cite[\S 10]{Distributed_Distinct_Elements-AFP} \\
Lem.~\ref{le:ind_bin_balls_exp} & \textbf{lemma} \emph{exp{\isacharunderscore}balls{\isacharunderscore}and{\isacharunderscore}bins} & \isatheory{DDE_Balls_And_Bins} & \cite[\S 4]{Distributed_Distinct_Elements-AFP} \\
Lem.~\ref{le:ind_bin_balls_var} & \textbf{lemma} \emph{var{\isacharunderscore}balls{\isacharunderscore}and{\isacharunderscore}bins} & \isatheory{DDE_Balls_And_Bins} & \cite[\S 4]{Distributed_Distinct_Elements-AFP} \\
Lem.~\ref{le:approx_bin_balls} (1) & \textbf{lemma} \emph{exp{\isacharunderscore}approx} & \isatheory{DDE_Balls_And_Bins} & \cite[\S 4]{Distributed_Distinct_Elements-AFP} \\ 
Lem.~\ref{le:approx_bin_balls} (2) & \textbf{lemma} \emph{var{\isacharunderscore}approx} & \isatheory{DDE_Balls_And_Bins} & \cite[\S 4]{Distributed_Distinct_Elements-AFP} \\ 
Lem.~\ref{le:balls_and_bins_sum} & \textbf{lemma} \emph{deviation{\isacharunderscore}bound} & \isatheory{DDE_Balls_And_Bins} & \cite[\S 4]{Distributed_Distinct_Elements-AFP} \\
\bottomrule
\\
\caption{Reference to the formal entities}\label{tab:formalization}
\end{longtable}

\end{document}